\documentclass[11pt]{article}

\usepackage{geometry}
\usepackage[utf8]{inputenc}
\usepackage[english]{babel}
\usepackage{lmodern}
\usepackage{microtype}

\usepackage{amsmath,amsfonts,amssymb}
\usepackage{amsthm}

\usepackage{graphicx}
\usepackage{authblk}
\usepackage{standalone}

\usepackage{hyperref}
\usepackage{xcolor}
\usepackage{standalone}
\usepackage{adjustbox}
\usepackage[capitalise]{cleveref}
\usepackage[shortlabels]{enumitem}

\usepackage{braket}
\usepackage{wrapfig}

\usepackage[nocompress]{cite}

\newcommand{\N}{\ensuremath{\mathbb{N}}}
\newcommand{\Z}{\ensuremath{\mathbb{Z}}}

\newcommand{\R}{\ensuremath{\mathbb{R}}}
\newcommand{\C}{\ensuremath{\mathbb{C}}}

\newcommand{\A}{\ensuremath{\mathcal{A}}}

\renewcommand{\epsilon}{\varepsilon}

\renewcommand{\vec}[1]{\mathbf{#1}}

\newcommand{\tr}{\ensuremath{ \mathrm{tr} }}

\newcommand{\wang}[6]{
\draw [black,fill=#3] (#1,#2+1)  -- (#1+0.5,#2+0.5) -- (#1+1,#2+1) -- cycle;
\draw [black,fill=#4] (#1+1,#2+1)  -- (#1+0.5,#2+0.5) -- (#1+1,#2) -- cycle;
\draw [black,fill=#5] (#1,#2)  -- (#1+0.5,#2+0.5) -- (#1+1,#2) -- cycle;
\draw [black,fill=#6] (#1,#2)  -- (#1+0.5,#2+0.5) -- (#1,#2+1) -- cycle;
}

\newtoggle{extern}
\togglefalse{extern}

\usepackage{stmaryrd}
\usepackage{mathtools}
\usepackage{old-arrows}
\usepackage{tikzit}
\usepackage{comment}

\tikzstyle{gn}=[fill=green, draw=black, shape=circle, tikzit category=ZX, tikzit fill=green, tikzit draw=black, tikzit shape=circle, inner sep=0.1em]
\tikzstyle{rn}=[fill=red, draw=black, shape=circle, tikzit fill=red, tikzit draw=black, tikzit category=ZX, tikzit shape=circle, inner sep=0.1em]
\tikzstyle{divide}=[regular polygon, regular polygon sides=3, shape border rotate=90, draw=black, fill={zx_grey}, inner sep=1.5pt, tikzit category=scal, rounded corners=0.8mm]
\tikzstyle{black}=[fill=black, draw=black, shape=circle, tikzit fill=black, tikzit draw=black, tikzit shape=circle, tikzit category=IH, inner sep=2pt]
\tikzstyle{gather}=[fill={zx_grey}, draw=black, tikzit category=scal, rounded corners=0.8mm, regular polygon, regular polygon sides=3, shape border rotate=-90, inner sep=1.5pt]
\tikzstyle{ggen}=[fill=white, draw=black, shape=rectangle, rounded corners=2mm, line width=1pt, tikzit draw=red, tikzit category=scal]
\tikzstyle{white}=[fill=white, draw=black, shape=circle, inner sep=2pt, tikzit category=IH]
\tikzstyle{mbox}=[fill=white, draw=black, rounded rectangle, rounded rectangle west arc=none, tikzit category=scal, tikzit shape=rectangle]
\tikzstyle{A}=[fill=white, shape=circle, tikzit category=scal, inner sep=1pt]
\tikzstyle{ggreen}=[fill=green, draw=black, shape=circle, tikzit category=SZX, tikzit fill=green, tikzit draw=black, line width=1pt, inner sep=0.1em]
\tikzstyle{gred}=[fill=red, draw=black, shape=circle, rounded corners=2mm, tikzit category=SZX, inner sep=0.1em, tikzit fill=red, line width=1pt]
\tikzstyle{ghad}=[minimum size=3mm, font={\scriptsize\boldmath}, shape=rectangle, inner sep=1mm, line width=1pt, outer sep=-1.5mm, scale=0.8, tikzit shape=rectangle, draw=black, fill=yellow, tikzit draw=blue]
\tikzstyle{boxm}=[fill=white, draw=black, rounded rectangle, tikzit category=scal, tikzit shape=rectangle, rounded rectangle east arc=none]
\tikzstyle{box}=[fill=white, draw=black, shape=rectangle]
\tikzstyle{had}=[fill=yellow, draw=black, shape=rectangle, tikzit category=ZX, tikzit fill=yellow, tikzit draw=black, inner sep=2.5pt]
\tikzstyle{gwhite}=[fill=white, draw=black, shape=circle, tikzit fill=white, tikzit shape=circle, line width=1 pt, inner sep=2 pt, tikzit draw=red]
\tikzstyle{gblack}=[fill=black, draw=black, shape=circle, tikzit fill=black, tikzit shape=circle, line width=1 pt, inner sep=2 pt, tikzit draw=red]
\tikzstyle{antipode}=[fill=red, draw=black, shape=rectangle, tikzit fill=red, tikzit draw=black, tikzit shape=rectangle, inner sep=2pt]
\tikzstyle{diamond}=[fill=white, draw=black, shape=diamond, inner sep=2pt]
\tikzstyle{mongr}=[fill=green, draw=green, shape=circle, inner sep=2pt]
\tikzstyle{monbl}=[fill=blue, draw=blue, shape=circle, inner sep=2pt]
\tikzstyle{bg}=[inner sep=0.7mm, minimum width=0pt, minimum height=0pt, fill=green, draw=white, very thick, shape=circle]
\tikzstyle{br}=[inner sep=0.7mm, minimum width=0pt, minimum height=0pt, fill=red, draw=white, very thick, shape=circle]
\tikzstyle{rmat}=[draw, signal, fill=red, signal to=east, signal from=west, inner sep=1pt, minimum height=6pt]
\tikzstyle{lmat}=[draw, signal, fill=red, signal to=west, signal from=east, inner sep=1pt, minimum height=6pt]
\tikzstyle{umat}=[draw, signal, fill=red, signal to=north, signal from=south, inner sep=1pt, minimum width=6pt]
\tikzstyle{dmat}=[draw, signal, fill=red, signal to=south, signal from=north, inner sep=1pt, minimum width=6pt]
\tikzstyle{box}=[shape=rectangle, text height=1.5ex, text depth=0.25ex, yshift=0.5mm, fill=white, draw=black, minimum height=5mm, yshift=-0.5mm, minimum width=5mm, font={\small}]
\tikzstyle{Z dot}=[inner sep=0mm, minimum size=2mm, shape=circle, draw=black, fill={rgb,255: red,160; green,255; blue,160}]
\tikzstyle{gdot}=[minimum size=3mm, font={\scriptsize\boldmath}, shape=rectangle, rounded corners=1.3mm, inner sep=1mm, outer sep=-1.8mm, scale=0.8, tikzit shape=circle, draw=black, fill=green, tikzit draw=blue]
\tikzstyle{X dot}=[Z dot, shape=circle, draw=black, fill={rgb,255: red,220; green,0; blue,0}]
\tikzstyle{rdot}=[minimum size=3mm, font={\scriptsize\boldmath}, shape=rectangle, rounded corners=1.3mm, inner sep=1mm, outer sep=-1.8mm, scale=0.8, tikzit shape=circle, draw=black, fill=red, tikzit draw=blue]
\tikzstyle{grdot}=[minimum size=3mm, font={\scriptsize\boldmath}, shape=rectangle, rounded corners=1.3mm, inner sep=1mm, line width=1pt, outer sep=-1.5mm, scale=0.8, tikzit shape=circle, draw=black, fill=red, tikzit draw=blue]
\tikzstyle{ggdot}=[minimum size=3mm, font={\scriptsize\boldmath}, shape=rectangle, line width=1pt, rounded corners=1.3mm, inner sep=1mm, outer sep=-1.5mm, scale=0.8, tikzit shape=circle, draw=black, fill=green, tikzit draw=blue]
\tikzstyle{arrow}=[-->]
\tikzstyle{rfarr}=[draw, signal, fill=black, signal to=east, signal from=west, inner sep=1pt, minimum height=6pt]
\tikzstyle{lfarr}=[draw, signal, fill=black, signal to=west, signal from=east, inner sep=1pt, minimum height=6pt]
\tikzstyle{ufarr}=[draw, signal, fill=black, signal to=north, signal from=south, inner sep=1pt, minimum width=6pt]
\tikzstyle{dfarr}=[draw, signal, fill=black, signal to=south, signal from=north, inner sep=1pt, minimum width=6pt]
\tikzstyle{ry}=[draw, signal, fill=yellow, signal to=east, signal from=west, inner sep=1pt, minimum height=6pt]
\tikzstyle{ly}=[draw, signal, fill=yellow, signal to=west, signal from=east, inner sep=1pt, minimum height=6pt]
\tikzstyle{uy}=[draw, signal, fill=yellow, signal to=north, signal from=south, inner sep=1pt, minimum width=6pt]
\tikzstyle{dy}=[draw, signal, fill=yellow, signal to=south, signal from=north, inner sep=1pt, minimum width=6pt]
\tikzstyle{red text}=[fill=none, text={rgb,255: red,177; green,37; blue,39}]
\tikzstyle{grey text}=[fill=none, draw=none, text={rgb,255: red,121; green,121; blue,121}]
\tikzstyle{diamond}=[fill=white, draw=black, shape=diamond]

\tikzstyle{arrow}=[->]
\tikzstyle{very thick}=[-, line width=1pt, tikzit draw=red]
\tikzstyle{pointille}=[dashed, -, draw={rgb,255: red,121; green,121; blue,121}, thin]
\tikzstyle{red}=[-, draw=red, opacity=0.5]
\tikzstyle{blue}=[-, draw=blue]
\tikzstyle{green}=[-, draw=green]
\tikzstyle{strike}=[-, tikzit draw={rgb,255: red,191; green,0; blue,64}, strike through]
\tikzstyle{strike'}=[-, tikzit draw=cyan, strike bend]
\tikzstyle{dashed arrow}=[->, tikzit draw=green, draw=black, dashed]
\tikzstyle{reprise}=[-, line width=2pt, tikzit draw={rgb,255: red,255; green,191; blue,191}]
\tikzstyle{grey}=[-, draw={rgb,255: red,191; green,191; blue,191}, opacity=0.5, fill={rgb,255: red,191; green,191; blue,191}]
\tikzstyle{red arrow}=[->, draw={rgb,255: red,177; green,37; blue,39}, fill=none, line width=0.5mm]
\tikzstyle{light red arrow}=[->, draw={rgb,255: red,177; green,37; blue,39}, fill=none, line width=0.5mm, opacity=0.70]
\tikzstyle{black}=[-, draw={rgb,255: red,21; green,21; blue,21}]
\tikzstyle{gris transparent}=[-, draw={rgb,255: red,121; green,121; blue,121}, thin, opacity=0.5]
\tikzstyle{grosse arrow}=[->, thick]

\newcommand{\bvdots}{ \tikz[baseline, every node/.style={inner sep=0}]{ \node at (0,0){.}; \node at (0,-6pt){.}; \node at (0,6pt){.}; } }

\newtheorem{open}{Open problem}

\theoremstyle{plain}
\newtheorem{theorem}{Theorem}[section]
\newtheorem{lemma}[theorem]{Lemma}

\newtheorem{proposition}[theorem]{Proposition}
\crefname{proposition}{Proposition}{Propositions}

\theoremstyle{definition}
\newtheorem{definition}{Definition}[section]
\theoremstyle{remark}
\newtheorem{remark}[theorem]{Remark}
\newtheorem{example}[theorem]{Example}

\title{Aperiodicity in Quantum Wang Tilings} 

\author{Titouan Carette}
\affil{\href{mailto:titouan.carette@lu.lv}{titouan.carette@polytechnique.edu}\\
  LIX, CNRS, École polytechnique, Institut Polytechnique de Paris}

\author{Etienne Moutot}
\affil{\href{mailto:etienne.moutot@math.cnrs.fr}{etienne.moutot@math.cnrs.fr}\\
CNRS, I2M, Aix-Marseille Université, Marseille, France}

\begin{document}

\maketitle

\begin{abstract}
  By reformulating Wang tiles with tensors, we propose a natural generalization to the probabilistic and quantum setting. In this new framework, we introduce notions of tilings and periodicity directly extending their classical counterparts. In the one dimensional case, we recover the decidability of the generalized domino problem by linking it to the trace characterization of nilpotent matrices. In the two-dimensional case, we provide extension of weak and strong aperiodicity respectively and show the equivalence of those generalized notions, extending the well known equivalence in the classical case. We also exhibit a quantum tile set being aperiodic while its underlying classical tile set is not, proving that quantum interference can suppress periodic patterns and paving the way to the investigation of a new kind of aperiodicity.
  Finally, we highlight the many new research directions opened by this generalization of Wang tiles, related to (quantum) cellular automata, condensed matter physics, symbolic dynamics and more. 
\end{abstract}

\section{Introduction}
  Wang tiles are one of the simplest tiling models one can think of. Each tile is a unit square with colored edges, and two tiles can be placed next to each other if and only if the color of their neighboring edges are the same.
  These seemingly simple local matching rules may however create extremely complex global behaviors, for example they can be used to encode any Turing machine \cite{Berger}.
  The strong links between Wang tiles, symbolic dynamics and compatibility theory made them a well-studied model over the years.
  The most natural computational model to relate to Wang tilings may be cellular automata, as any space-time diagram of a cellular automata can be interpreted as a Wang tiling for a well-chosen set of Wang tiles.
  This strong connection led to many undecidability results being proved by reductions to tiling-related decision problems \cite{Kari_1992, Kari_1994, Kari_Ollinger_2008}.
  As many other computational systems, cellular automata have seen their quantum version studied thoroughly, even started from Feynmann \cite{Feynman_1982}. Coming up with the right formal model took a lot of time and effort \cite{Watrous_1995, Schumacher_Werner_2004, Arrighi_Nesme_Werner_2011}, but eventually led to a fruitful discrete model of physics.
  For now, it seems unlikely that they offer any computational advantage compared to other quantum computational models. However, they seem to provide a fantastic tool to simulate other quantum systems \cite{Arrighi_Beny_Farrelly_2020, Farrelly_Streich_2020}.

  Surprisingly, there exist no quantum version of tiling models that we are aware of, even the simplest Wang tile model, that would play the role of space-time diagrams for quantum cellular automata.
  We propose a formalism for \emph{quantum wang tilesets}, and we look into the details of their similarities and differences with their classical counterpart.

  To each tile in a tile set we assign an amplitude, a complex number, from which we can then compute an amplitude for any pattern tiling a given shape. The probability of observing this pattern is the square of the modulus of the amplitude: $P(c)=|a_c |^2 $. This implies that the sums of the modulus squared of the amplitudes of all possible patterns must be one: $\sum_c |a_c |^2 = 1$. Interference occurs when we compute the amplitude of events involving multiple patterns. Then, the rules of quantum mechanics impose that the amplitude of the event is the sum of the amplitudes of the configurations involved: $P(E)=|a_E |^2 = |\sum_{c\in E} a_c |^2 $. Hence, the amplitudes being complex numbers, it is possible that such sum is zero, leading to a counterintuitive, yet experimentally observed, situation where the combination of independently valid patterns leads to a never observed event. Another important consequence follows from the non-copy theorem. As tiles need to send information about themselves to their neighbours, we can only extract information from tiles on the boundary. Thus we will only be able to compute probabilities of events happening on the periphery of the tiled shape. Obtaining internal information would require measuring between the tiles and then preventing interferences. 

  One of the essential properties of tilings is \emph{periodicity}. A key element of the expressiveness of tilings as a computational model is the existence of aperiodic tilesets: a set of tiles that does tile the plane but only in a non-periodic manner \cite{Berger}.
  Periodicity is such a fundamental property of tilings that it seems natural to investigate it in our new quantum setting, where interference phenomena have exciting consequences.

  In this paper, we define a new model of quantum Wang dominoes (dimension 1) and tiles (dimension 2) that can be easily generalized in any dimension.
  In dimension 1, our model is represented by a matrix whose complex coefficients encode the amplitude of each tile being valid. The classical counterpart is a matrix whose 0-1 coefficients encode the fact that a tile is present or not in the tileset. In this case, the coefficients of the $n$-th power of the matrix encode the number of valid patterns of size $n$. In higher dimension, the matrix becomes a tensor and the matrix product is generalized by a tensor contraction, but the interpretation of the coefficients remains the same.
  A similar encoding of classical Wang tiles into tensors has been used in \cite{Tenseur2020, Tenseur2022} to perform reduction decidability of the nullity of a tensor network to the domino problem. 

  Already in dimension 1, quantum interference make the notion of \textit{tiling the line} more involved than in the classical case. Indeed, we show that destructive interference can lead to a tileset that do not tile the line at large scale, even if the underlying classical set of dominoes does.
  In dimension 2, a similar phenomenon happens with aperiodicity: we demonstrate that a new \emph{kind} of aperiodicity may appear, where the periodic tilings are annihilated by destructive interference.
  In order to do that, we also generalize several usual properties of (a)periodic tilings in the quantum setting, which is made non-trivial by the presence of interference.


  The paper is organized as follows. The first section introduces the tensor formalism with a particular emphasis on abstract index notation and string diagrams that are used everywhere after.
  In \cref{sec:1D} we gently introduce the model by studying one-dimensional dominoes and their properties, as well as giving example of interesting quantum behaviors.
  In \cref{sec:2D}, we define the 2D model of quantum Wang tiles and focus on properties of their (a)periodicity. We show that some classical periodicity properties still holds, but that new non-equivalent definitions arise, and provide an example of tileset illustrating this new ''quantum'' definition of aperiodicity.
  In \cref{sec:appli}, we give some examples of applications of this new model to quantum cellular automata and quantum walks.
  Finally, we conclude on the (many) possible research directions with this new model.

\section{Background on tensors} \label{sec:tensors}
  Before introducing our tensorial reformulation of Wang tiles, we review the needed notions and notations relative to tensors in this section.

  \subsection{Definitions}
    In the whole section, we fix a vector space $V$ of finite dimension $d$ over a field $\mathbb{K}$. We also fix a canonical basis $(v_i)_{1\leq i \leq d}$ of $V$. Given two vector spaces $V$ and $W$, of dimension respectively $d$ and $k$, with basis respectively $(v_i)_{1\leq i \leq d}$ and $({w}_j)_{1\leq i \leq k}$, their \textbf{tensor product} is a space $V\otimes W$ of dimension $d \times k$, a basis of which is given by $(v_i \otimes {w}_j )_{1\leq i \leq d,~1\leq j \leq k}$. So fixing a canonical basis for $V$ gives us directly a canonical basis for ${V }^{\otimes n} = V\otimes \cdots \otimes V$. By convention $V^{\otimes 0}=\mathbb{K}$ and $V^{\otimes 1}=V$.

    \begin{definition}[Tensor]
         An $n$-\textbf{tensor} is a vector $T \in {V }^{\otimes n}$.
    \end{definition}

    Since we fixed a basis, we can identify an $n$-tensor with a list of $d^n$ coefficients denoted: $T_{i_1 ~ \cdots~ i_m}\in \mathbb{K}$ indexed by $n$ indices $1\leq i_1 , \cdots, i_m, j_1 , \cdots, j_n \leq d$. When dealing with long lists of indices, we will write $\mathbf{i}$ in bold font instead of $i_1 ~ \cdots~ i_m$, furthermore the notation $1\leq\mathbf{i}\leq d$ indicates that the indices in the list $\mathbf{i}$ are integers between $1$ and $d$. The concatenation of two lists of indices is denoted $\mathbf{i},\mathbf{j}$.  The coefficients of an $n$-tensor are then denoted $T_{\mathbf{i}}\in \mathbb{K}$.

    \begin{example}
         There are numerous examples of tensors:
         \begin{itemize}
              \item A $0$-tensor is a scalar in $\mathbb{K}$.
              \item A $1$-tensor is a vector $v\in V$.
              \item A $2$-tensor can be identified with a linear map $A:V\to V$, the coefficients of the associated matrix being exactly the coefficient $A_{i,j}$ of the tensor.
              \item More generally we can always partition the indices in two sets and interpret any tensor as a linear map $V^{\otimes n}\to V^{\otimes m}$.
         \end{itemize}
    \end{example}
	
	A more concrete way to understand tensors is to see them as multi indexed tables of coefficients, a $0$-tensor is just one number, a $1$-tensor is a list of numbers, a $2$-tensor is a matrix of numbers, a $3$-tensor is a cube \textit{etc...}
	
	\begin{remark}
		The reader familliar with tensors will notice that we are only considering covariant indices, indeed we will only work in finite dimension and with a fixed prefered basis, so we have a canonical isomorphism $V\simeq V^* $, allowing us to simplify the presentation.
	\end{remark}

    We can operate on tensors in various ways.
    \begin{definition}[Tensor product]
         Given a $n$-tensor $T$ and a $m$-tensor ${L}$, their tensor product is the  $(n+m)$-tensor $T\otimes L$ defined as: $(T\otimes L)_{\mathbf{i}, \mathbf{j}}= T_{\mathbf{i}}{L}_{\mathbf{j}}$ with $1\leq \mathbf{i}, \mathbf{j}\leq d$.
    \end{definition}
       \noindent
    This generalises the Kronecker product of matrices.

    \begin{definition}[Contraction]
         Given a $n$-tensor $T$ with $n\geq 2$ and choosing two different indices in positions $a$ and $b$, we can form a $(n-2)$-tensor $\mathrm{tr}_{a,b}(T)$ defined as: $\mathrm{tr}_{a,b} (T)_{\mathbf{i},\mathbf{j},\mathbf{k}}=\sum\limits_{1\leq \ell\leq d} T_{\mathbf{i},\ell,\mathbf{j},\ell,\mathbf{k}}$. Here the lists of indices $\mathbf{i}$, $\mathbf{j}$ and $\mathbf{k}$ are respectively of size $a-1$, $b-a-1$ and $n-b$. 
    \end{definition}

    If we contract a matrix $A$ seen as a $2$-tensor we should obtain a $0$-tensor, \textit{i.e.}, a scalar. In fact the formula gives us $\mathrm{tr}_{1,2}(A)=\sum_\ell A_{\ell,\ell}$, this is the trace of $A$, justifying the trace notation. One can remark that applying contraction to the tensor product of two $2$-tensors $A$ and $B$ seen as matrices is in fact the usual product of matrices. Indeed $\mathrm{tr}_{2,3}(A\otimes B)_{i,j}= \sum_\ell A_{i,\ell}B_{\ell,j}=(AB)_{i,j}$.

  \subsection{Abstract index notation}
    The \textbf{abstract index notation} consists in writing a $n$-tensor $T$ as $T_{x_1 , \cdots, x_m}$ here the $x_k$ must not be thought as number but as indeterminates in a similar way we write polynomials with indeterminates. We can then instantiates those indeterminates by concrete indices to obtain the coefficients of the tensor. This notation allows to represent operations on tensors in a very compact way. Using the bold font notation for lists of indices, the tensor product of two tensors $T_{\mathbf{x}}$ and ${L}_{\mathbf{y}}$ is directly written $T_{\mathbf{x}}{L}_{\mathbf{y}}$. The \textbf{Einstein summation convention} keeps some particular sums implicit. If an index is repeated, meaning we use the same indeterminate to denote two different indices, then it means that we sum over those indices. So a contracted tensor will be directly written $T_{\mathbf{x},t,\mathbf{y},t,\mathbf{z}}$ where the repeated $t$ are at the positions of the contracted indices. Using abstract index notation, we denote the permutation of indices by a permutation $\sigma\in \mathfrak{S}_n$ as $T_{\sigma(\mathbf{x})}$.

  \subsection{Penrose notations}
    Penrose notations, introduced in \cite{penrose1971applications}, are the graphical counterpart of the abstract index notation. The idea is to represent tensors as boxes and indices as wires. A $n$-tensor is then represented by a box with $n$ dangling wires.
    \begin{center}
        \begin{tikzpicture}
	\begin{pgfonlayer}{nodelayer}
		\node [style=none] (0) at (0.5, 0) {$T$};
		\node [style=none] (1) at (-0.25, 0.75) {};
		\node [style=none] (2) at (1.25, 0.75) {};
		\node [style=none] (3) at (1.25, -0.75) {};
		\node [style=none] (4) at (-0.25, -0.75) {};
		\node [style=none] (6) at (0.5, -1) {$...$};
		\node [style=none] (11) at (1, -0.75) {};
		\node [style=none] (12) at (1, -1.25) {};
		\node [style=none] (13) at (0, -0.75) {};
		\node [style=none] (14) at (0, -1.25) {};
		\node [style=none] (15) at (-2, 0) {$T_{\mathbf{x}}$};
		\node [style=none] (16) at (6.25, 0) {$T$};
		\node [style=none] (17) at (5.5, 0.75) {};
		\node [style=none] (18) at (7, 0.75) {};
		\node [style=none] (19) at (7, -0.75) {};
		\node [style=none] (20) at (5.5, -0.75) {};
		\node [style=none] (22) at (6.25, -1) {$...$};
		\node [style=none] (27) at (6.75, -0.75) {};
		\node [style=none] (28) at (6.75, -1.25) {};
		\node [style=none] (29) at (5.75, -0.75) {};
		\node [style=none] (30) at (5.75, -1.25) {};
		\node [style=none] (31) at (3.25, 0) {$T_{\mathbf{x}}{L}_{\mathbf{y}}$};
		\node [style=none] (32) at (8.25, 0) {$L$};
		\node [style=none] (33) at (7.5, 0.75) {};
		\node [style=none] (34) at (9, 0.75) {};
		\node [style=none] (35) at (9, -0.75) {};
		\node [style=none] (36) at (7.5, -0.75) {};
		\node [style=none] (38) at (8.25, -1) {$...$};
		\node [style=none] (43) at (8.75, -0.75) {};
		\node [style=none] (44) at (8.75, -1.25) {};
		\node [style=none] (45) at (7.75, -0.75) {};
		\node [style=none] (46) at (7.75, -1.25) {};
		\node [style=none] (74) at (15.25, 0) {$T$};
		\node [style=none] (75) at (16.25, 0.75) {};
		\node [style=none] (76) at (14.25, 0.75) {};
		\node [style=none] (77) at (14.25, -0.75) {};
		\node [style=none] (78) at (16.25, -0.75) {};
		\node [style=none] (80) at (15.5, -1) {$...$};
		\node [style=none] (85) at (14.5, -0.75) {};
		\node [style=none] (86) at (14.5, -1.25) {};
		\node [style=none] (87) at (16, -0.75) {};
		\node [style=none] (88) at (16, -1.25) {};
		\node [style=none] (89) at (11.5, 0) {$\mathrm{tr}_{1,2}(T)_{\mathbf{x}}$};
		\node [style=none] (101) at (15, -1.25) {};
		\node [style=none] (105) at (14.5, -1.25) {};
		\node [style=none] (107) at (14.75, -1.5) {};
		\node [style=none] (108) at (21.25, 0) {$T$};
		\node [style=none] (109) at (20.25, 0.75) {};
		\node [style=none] (110) at (22.25, 0.75) {};
		\node [style=none] (111) at (22.25, -0.75) {};
		\node [style=none] (112) at (20.25, -0.75) {};
		\node [style=none] (121) at (20.75, -0.75) {};
		\node [style=none] (122) at (21.25, -1.5) {};
		\node [style=none] (123) at (21.75, -1.5) {};
		\node [style=none] (124) at (21.25, -0.75) {};
		\node [style=none] (125) at (21.75, -0.75) {};
		\node [style=none] (126) at (20.75, -1.5) {};
		\node [style=none] (127) at (18, 0) {$T_{\sigma(\mathbf{x})}$};
		\node [style=none] (128) at (-1, 0) {$=$};
		\node [style=none] (129) at (4.75, 0) {$=$};
		\node [style=none] (131) at (13.5, 0) {$=$};
		\node [style=none] (132) at (19.5, 0) {$=$};
		\node [style=none] (133) at (15, -0.75) {};
		\node [style=none] (134) at (15, -1.25) {};
	\end{pgfonlayer}
	\begin{pgfonlayer}{edgelayer}
		\draw (11.center) to (12.center);
		\draw (14.center) to (13.center);
		\draw (1.center) to (4.center);
		\draw (4.center) to (3.center);
		\draw (2.center) to (3.center);
		\draw (2.center) to (1.center);
		\draw (27.center) to (28.center);
		\draw (30.center) to (29.center);
		\draw (17.center) to (20.center);
		\draw (20.center) to (19.center);
		\draw (18.center) to (19.center);
		\draw (18.center) to (17.center);
		\draw (43.center) to (44.center);
		\draw (46.center) to (45.center);
		\draw (33.center) to (36.center);
		\draw (36.center) to (35.center);
		\draw (34.center) to (35.center);
		\draw (34.center) to (33.center);
		\draw (85.center) to (86.center);
		\draw (88.center) to (87.center);
		\draw (75.center) to (78.center);
		\draw (78.center) to (77.center);
		\draw (76.center) to (77.center);
		\draw (76.center) to (75.center);
		\draw [bend right=45] (105.center) to (107.center);
		\draw [bend right=45] (107.center) to (101.center);
		\draw (109.center) to (112.center);
		\draw (112.center) to (111.center);
		\draw (110.center) to (111.center);
		\draw (110.center) to (109.center);
		\draw [in=-90, out=90] (122.center) to (121.center);
		\draw [in=90, out=-90] (124.center) to (123.center);
		\draw [in=-90, out=90, looseness=0.75] (126.center) to (125.center);
		\draw (133.center) to (134.center);
	\end{pgfonlayer}
\end{tikzpicture}
    \end{center}
    We can then draw all tensor operations by directly mimicking the abstract index notation. The repeated indices correspond to links. Crossings of wires represent permutations of indices. Notice that the $2$-tensor corresponding to the identity matrix is denoted simply as a wire.
    
    We can rigorously formalize all those notations and the corresponding diagrammatical equational theory in monoidal categories. We invite the interested reader to refer to \cite{selinger2010survey}. We will extensively use those different notations in our tensorial take on Wang tiles. Notice that we can be very lax with our way of representing tensor once we agree on which wire represent which index. 
    For example, a tensor with $4$ indices will be freely represented as:
    \begin{center}
        \begin{tikzpicture}
	\begin{pgfonlayer}{nodelayer}
		\node [style=none] (169) at (1, 0) {$T$};
		\node [style=none] (170) at (0, 0.75) {};
		\node [style=none] (171) at (2, 0.75) {};
		\node [style=none] (172) at (2, -0.75) {};
		\node [style=none] (173) at (0, -0.75) {};
		\node [style=none] (176) at (0.5, 0.75) {};
		\node [style=none] (177) at (0.5, 1.25) {};
		\node [style=none] (178) at (1.5, 1.25) {};
		\node [style=none] (179) at (1.5, 0.75) {};
		\node [style=none] (180) at (1.5, -0.75) {};
		\node [style=none] (181) at (1.5, -1.25) {};
		\node [style=none] (182) at (0.5, -0.75) {};
		\node [style=none] (183) at (0.5, -1.25) {};
		\node [style=none] (189) at (12, 0.25) {$T$};
		\node [style=none] (190) at (10, -0.25) {};
		\node [style=none] (191) at (14, -0.25) {};
		\node [style=none] (192) at (14, 0.75) {};
		\node [style=none] (193) at (10, 0.75) {};
		\node [style=none] (196) at (10.5, -0.25) {};
		\node [style=none] (197) at (10.5, -0.75) {};
		\node [style=none] (198) at (11.5, -0.75) {};
		\node [style=none] (199) at (11.5, -0.25) {};
		\node [style=none] (200) at (13.5, -0.75) {};
		\node [style=none] (201) at (13.5, -0.25) {};
		\node [style=none] (202) at (12.5, -0.75) {};
		\node [style=none] (203) at (12.5, -0.25) {};
		\node [style=none] (208) at (6, 0) {$T$};
		\node [style=none] (209) at (6, 1) {};
		\node [style=none] (210) at (7, 0) {};
		\node [style=none] (211) at (6, -1) {};
		\node [style=none] (212) at (5, 0) {};
		\node [style=none] (213) at (5.5, 0.5) {};
		\node [style=none] (214) at (5, 1) {};
		\node [style=none] (215) at (7, 1) {};
		\node [style=none] (216) at (6.5, 0.5) {};
		\node [style=none] (217) at (5.5, -0.5) {};
		\node [style=none] (218) at (5, -1) {};
		\node [style=none] (219) at (6.5, -0.5) {};
		\node [style=none] (220) at (7, -1) {};
		\node [style=none] (221) at (18.5, 0) {$T$};
		\node [style=none] (222) at (17.75, 0.75) {};
		\node [style=none] (223) at (19.25, 0.75) {};
		\node [style=none] (224) at (19.25, -0.75) {};
		\node [style=none] (225) at (17.75, -0.75) {};
		\node [style=none] (226) at (18.5, 0.75) {};
		\node [style=none] (227) at (18.5, 1.25) {};
		\node [style=none] (228) at (19.25, 0) {};
		\node [style=none] (229) at (19.75, 0) {};
		\node [style=none] (230) at (18.5, -0.75) {};
		\node [style=none] (231) at (18.5, -1.25) {};
		\node [style=none] (232) at (17.75, 0) {};
		\node [style=none] (233) at (17.25, 0) {};
	\end{pgfonlayer}
	\begin{pgfonlayer}{edgelayer}
		\draw (177.center) to (176.center);
		\draw (179.center) to (178.center);
		\draw (180.center) to (181.center);
		\draw (183.center) to (182.center);
		\draw (170.center) to (173.center);
		\draw (173.center) to (172.center);
		\draw (171.center) to (172.center);
		\draw (171.center) to (170.center);
		\draw (197.center) to (196.center);
		\draw (199.center) to (198.center);
		\draw (200.center) to (201.center);
		\draw (203.center) to (202.center);
		\draw (190.center) to (193.center);
		\draw (193.center) to (192.center);
		\draw (191.center) to (192.center);
		\draw (191.center) to (190.center);
		\draw (214.center) to (213.center);
		\draw (216.center) to (215.center);
		\draw (217.center) to (218.center);
		\draw (220.center) to (219.center);
		\draw (209.center) to (212.center);
		\draw (212.center) to (211.center);
		\draw (210.center) to (211.center);
		\draw (210.center) to (209.center);
		\draw (227.center) to (226.center);
		\draw (229.center) to (228.center);
		\draw (230.center) to (231.center);
		\draw (233.center) to (232.center);
		\draw (222.center) to (225.center);
		\draw (225.center) to (224.center);
		\draw (223.center) to (224.center);
		\draw (223.center) to (222.center);
	\end{pgfonlayer}
\end{tikzpicture}
    \end{center}
    We can straightforwardly recover the proper abstract index notation from a diagram as long as we are clear on which link corresponds to which index. Typically, the tensor product of two $4$-tensor $T$ and $L$ and the tensor obtain by contracting an index of $T$ and an index of $L$ can be depicted respectively as:
    
    \begin{center}
    	\begin{tikzpicture}
	\begin{pgfonlayer}{nodelayer}
		\node [style=none] (208) at (1, 0) {$T$};
		\node [style=none] (209) at (1, 1) {};
		\node [style=none] (210) at (2, 0) {};
		\node [style=none] (211) at (1, -1) {};
		\node [style=none] (212) at (0, 0) {};
		\node [style=none] (213) at (0.5, 0.5) {};
		\node [style=none] (214) at (0, 1) {};
		\node [style=none] (215) at (2, 1) {};
		\node [style=none] (216) at (1.5, 0.5) {};
		\node [style=none] (217) at (0.5, -0.5) {};
		\node [style=none] (218) at (0, -1) {};
		\node [style=none] (219) at (1.5, -0.5) {};
		\node [style=none] (220) at (2, -1) {};
		\node [style=none] (221) at (4, 0) {$L$};
		\node [style=none] (222) at (4, 1) {};
		\node [style=none] (223) at (5, 0) {};
		\node [style=none] (224) at (4, -1) {};
		\node [style=none] (225) at (3, 0) {};
		\node [style=none] (226) at (3.5, 0.5) {};
		\node [style=none] (227) at (3, 1) {};
		\node [style=none] (228) at (5, 1) {};
		\node [style=none] (229) at (4.5, 0.5) {};
		\node [style=none] (230) at (3.5, -0.5) {};
		\node [style=none] (231) at (3, -1) {};
		\node [style=none] (232) at (4.5, -0.5) {};
		\node [style=none] (233) at (5, -1) {};
		\node [style=none] (234) at (11, 0.75) {$T$};
		\node [style=none] (235) at (11, 1.75) {};
		\node [style=none] (236) at (12, 0.75) {};
		\node [style=none] (237) at (11, -0.25) {};
		\node [style=none] (238) at (10, 0.75) {};
		\node [style=none] (239) at (10.5, 1.25) {};
		\node [style=none] (240) at (10, 1.75) {};
		\node [style=none] (241) at (12, 1.75) {};
		\node [style=none] (242) at (11.5, 1.25) {};
		\node [style=none] (243) at (10.5, 0.25) {};
		\node [style=none] (244) at (10, -0.25) {};
		\node [style=none] (247) at (12.5, -0.75) {$L$};
		\node [style=none] (248) at (12.5, 0.25) {};
		\node [style=none] (249) at (13.5, -0.75) {};
		\node [style=none] (250) at (12.5, -1.75) {};
		\node [style=none] (251) at (11.5, -0.75) {};
		\node [style=none] (252) at (12, -0.25) {};
		\node [style=none] (253) at (11.5, 0.25) {};
		\node [style=none] (254) at (13.5, 0.25) {};
		\node [style=none] (255) at (13, -0.25) {};
		\node [style=none] (256) at (12, -1.25) {};
		\node [style=none] (257) at (11.5, -1.75) {};
		\node [style=none] (258) at (13, -1.25) {};
		\node [style=none] (259) at (13.5, -1.75) {};
		\node [style=none] (260) at (7.5, 0) {and};
	\end{pgfonlayer}
	\begin{pgfonlayer}{edgelayer}
		\draw (214.center) to (213.center);
		\draw (216.center) to (215.center);
		\draw (217.center) to (218.center);
		\draw (220.center) to (219.center);
		\draw (209.center) to (212.center);
		\draw (212.center) to (211.center);
		\draw (210.center) to (211.center);
		\draw (210.center) to (209.center);
		\draw (227.center) to (226.center);
		\draw (229.center) to (228.center);
		\draw (230.center) to (231.center);
		\draw (233.center) to (232.center);
		\draw (222.center) to (225.center);
		\draw (225.center) to (224.center);
		\draw (223.center) to (224.center);
		\draw (223.center) to (222.center);
		\draw (240.center) to (239.center);
		\draw (242.center) to (241.center);
		\draw (243.center) to (244.center);
		\draw (235.center) to (238.center);
		\draw (238.center) to (237.center);
		\draw (236.center) to (237.center);
		\draw (236.center) to (235.center);
		\draw (253.center) to (252.center);
		\draw (255.center) to (254.center);
		\draw (256.center) to (257.center);
		\draw (259.center) to (258.center);
		\draw (248.center) to (251.center);
		\draw (251.center) to (250.center);
		\draw (249.center) to (250.center);
		\draw (249.center) to (248.center);
	\end{pgfonlayer}
\end{tikzpicture}
    \end{center}

\section{Tensorial dominoes} \label{sec:1D}
  We will start by stating our tensorial reformulation of Wang tiles in the one dimensional case. It will allow us to present in a simple setting the different definitions and subtleties of the formalism before moving to the two-dimensional case in the next section, where a lot more care is needed to define everything properly.
      
  \subsection{Definition} 
    We fix a finite set of colors $A$ and fix a basis of $\mathbb{C}^{|A|}$ indexed by $A$. Basis elements are denoted $|c\rangle$ with $c\in A$. This quantum mechanical notation is read "ket $c$". The conjugate transposed of such vectors is denoted $\langle c| = |c\rangle^{\dagger}$, which is read "bra $c$". The tensor product of two basis elements will be denoted: $|ab\rangle = |a\rangle|b\rangle = |a\rangle\otimes |b\rangle$. Notice that for any matrix $M: \mathbb{C}^{|A|} \to \mathbb{C}^{|A|}$ we have $M_{i,j}= \bra{j}M\ket{i}$. A one-dimensional $A$-colored Wang tile, or domino, is a couple $(a,b)\in A\times A$. A set of dominoes is then a subset $D\subseteq A\times A$. We will identify a domino set $D$ to a $2$-tensor $T\in \mathbb{C}^{|A|} \otimes \mathbb{C}^{|A|}$ defined as $T= \sum\limits_{(x,y)\in D} |xy\rangle $ whose coefficients are then either $0$ or $1$.

    \begin{definition}
      A \textbf{possibilistic domino} is a $2$-tensor $T\in \mathbb{C}^{|A|} \otimes \mathbb{C}^{|A|}$ whose coefficients are either $0$ or $1$.
    \end{definition}

    Notice that possibilistic dominos are in bijection with domino sets. In general we will call \textbf{tensorial domino} any tensor $T$ with two indices. The \textbf{support} of $T$ is the classical domino set $\mathrm{supp}(T)$ defined as $\mathrm{supp}(T)= \{(x,y)\in A^2, T_{x,y}\neq 0 \}$. This abstract definition will allow us to extend the domino sets to the probabilistic and quantum setting.

    \begin{definition}
      A \textbf{probabilistic domino} is a $2$-tensor $T\in\mathbb{C}^{|A|} \otimes \mathbb{C}^{|A|}$ whose coefficients are in $[0,1]$ and such that $\sum\limits_{(x,y)\in A^2 } T_{x,y}=1 $.
    \end{definition}

    A probabilistic domino $T$ is the same as a probability distribution over $\mathrm{supp}(T)$. It can be interpreted as a random generator outputting dominoes according to some distribution.

    \begin{definition}
      A \textbf{quantum domino} is a $2$-tensor $T\in \mathbb{C}^{|A|} \otimes \mathbb{C}^{|A|}$ such that $\sum\limits_{(x,y)\in A^2 } |T_{x,y}|^2 =1 $.
    \end{definition}

    Equivalently a quantum domino can be seen as a quantum state where all dominoes of the support are in superposition. The coefficients of the tensor are the amplitudes assigned to each domino.


    Tensorial dominoes can be combined. To simplify the description of those compositions, we will see tensorial dominoes as matrices with coefficients $T_{x,y}$.

    \begin{definition}[Product and union]
      The \textbf{product}, respectively the \textbf{union}, of two tensorial dominoes $T\in \mathcal{M}_{d\times d}(\mathbb{C})$ and $L\in \mathcal{M}_{k\times k}(\mathbb{C})$ are respectively defined as: $(T\times L) = T\otimes L \in \mathcal{M}_{dk \times dk }(\mathbb{C})$ and $(T\uplus L) = T\oplus L \in \mathcal{M}_{(d+k) \times (d+k)}(\mathbb{C})$. Where $\otimes$ is the Kronecker tensor of matrices and $\oplus $ is the direct product of matrices. Notice that those tensorial dominoes need not to be definied on the same alphabet.
    \end{definition}

    Those operations correspond respectively to the cartesian product and disjoint union of both supports, we have $ \mathrm{supp}(T\times T')=\mathrm{supp}(T) \times \mathrm{supp}(T')$ and $ \mathrm{supp}(T\uplus T')= \mathrm{supp}(T)\uplus \mathrm{supp}(T')$.

  \subsection{Tilings as matrix product}
    When it comes to considering the ways to tile the line using tensorial dominoes, we will use the matrix interpretation of $2$-tensors: $T= \sum\limits_{(x,y)\in A^2 } T_{x,y} |x\rangle\langle y| $ depicted:

    \begin{center}
      \begin{tikzpicture}
	\begin{pgfonlayer}{nodelayer}
		\node [style=none] (498) at (1.75, 0.75) {};
		\node [style=none] (499) at (3.25, 0.75) {};
		\node [style=none] (500) at (3.25, -0.75) {};
		\node [style=none] (501) at (1.75, -0.75) {};
		\node [style=none] (502) at (2.5, 0) {$T$};
		\node [style=none] (507) at (1.5, 0) {};
		\node [style=none] (508) at (1.75, 0) {};
		\node [style=none] (511) at (3.25, 0) {};
		\node [style=none] (512) at (3.5, 0) {};
		\node [style=none] (575) at (-1, 0) {$T_{x,y}$};
		\node [style=none] (577) at (1, 0) {$x$};
		\node [style=none] (578) at (4, 0) {$y$};
		\node [style=none] (579) at (0, 0) {$=$};
	\end{pgfonlayer}
	\begin{pgfonlayer}{edgelayer}
		\draw (498.center) to (499.center);
		\draw (500.center) to (499.center);
		\draw (498.center) to (501.center);
		\draw (501.center) to (500.center);
		\draw (511.center) to (512.center);
		\draw (508.center) to (507.center);
	\end{pgfonlayer}
\end{tikzpicture}
    \end{center}

    Then we will benefit from the historical bias toward $2$-tensors and use results and notations from familiar matrix algebra.

    Let's look closely at the formula of the composition of two copies of a tensorial domino $T$: $(T^2)_{x,y}=\sum_k T_{x,k} T_{k,y}$ depicted as:

    \begin{center}
      \begin{tikzpicture}
	\begin{pgfonlayer}{nodelayer}
		\node [style=none] (498) at (0.5, 0.75) {};
		\node [style=none] (499) at (2, 0.75) {};
		\node [style=none] (500) at (2, -0.75) {};
		\node [style=none] (501) at (0.5, -0.75) {};
		\node [style=none] (502) at (1.25, 0) {$T$};
		\node [style=none] (507) at (0, 0) {};
		\node [style=none] (508) at (0.5, 0) {};
		\node [style=none] (513) at (3, 0.75) {};
		\node [style=none] (514) at (4.5, 0.75) {};
		\node [style=none] (515) at (4.5, -0.75) {};
		\node [style=none] (516) at (3, -0.75) {};
		\node [style=none] (517) at (3.75, 0) {$T$};
		\node [style=none] (518) at (2, 0) {};
		\node [style=none] (519) at (3, 0) {};
		\node [style=none] (520) at (4.5, 0) {};
		\node [style=none] (521) at (5, 0) {};
	\end{pgfonlayer}
	\begin{pgfonlayer}{edgelayer}
		\draw (498.center) to (499.center);
		\draw (500.center) to (499.center);
		\draw (498.center) to (501.center);
		\draw (501.center) to (500.center);
		\draw (508.center) to (507.center);
		\draw (513.center) to (514.center);
		\draw (515.center) to (514.center);
		\draw (513.center) to (516.center);
		\draw (516.center) to (515.center);
		\draw (520.center) to (521.center);
		\draw (519.center) to (518.center);
	\end{pgfonlayer}
\end{tikzpicture}
    \end{center}

    The matrix $T^2 $'s coefficients are the number of admissible length $2$ patterns from the support domino set. By admissible, we mean that the two colors in the middle match. The coefficients of $T^n $ can then be interpreted as follow for $a$ and $b$ two colors of $A$:

    \begin{itemize}
    \item If $T$ is possibilistic, then $\langle a|T^{n}|b\rangle$ is the number of admissible length $n$ patterns, starting with color $a$ and ending with color $b$, made from dominoes in the support of $T$.
    \item If $T$ is probabilistic, then $\langle a|T^{n}|b\rangle$ is the probability of forming an admissible length $n$ patterns, starting with color $a$ and ending with color $b$, by sampling $n$ dominoes at random in the support.
    \item If $T$ is quantum, then $\langle a|T^{n}|b\rangle$ is the complex amplitude corresponding of the event of forming an admissible length $n$ patterns, starting with color $a$ and ending with color $b$, by making $n$ quantum dominoes interact.
    \end{itemize}

    \begin{example}\label{ex:1Dclassique}
     Let $T = \begin{pmatrix}
     1&1 \\ 1&0
     \end{pmatrix}$. It is a possibilistic domino representing the classical set of dominoes $\{ \includegraphics[scale=0.5]{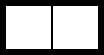},\includegraphics[scale=0.5]{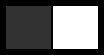},\includegraphics[scale=0.5]{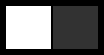} \}$.
     The fact that $T^2 = \begin{pmatrix}
     2&1 \\ 1&1
     \end{pmatrix}$ shows that there are two tilings starting and ending with white, which are \includegraphics[scale=0.5]{white-white}\includegraphics[scale=0.5]{white-white} and \includegraphics[scale=0.5]{white-black}\includegraphics[scale=0.5]{black-white}; and then one for each other borders: \includegraphics[scale=0.5]{black-white}\includegraphics[scale=0.5]{white-white}, \includegraphics[scale=0.5]{white-white}\includegraphics[scale=0.5]{white-black} and \includegraphics[scale=0.5]{black-white}\includegraphics[scale=0.5]{white-black}.
    \end{example}

    In usual terms, a set of dominoes $\tau$ \emph{tiles the line $\Z$} if there exists a configuration $x\in \tau^\Z$ such that each color of the dominoes ``match their neighbors'': $\forall i\in\Z,$ there exists $a,b,c,d\in A$ such that $x_{i-1}=(a,b), x_i=(b,c)$ and $x_{i+1} = (c,d)$.
    A \emph{valid pattern} is an assamblage of tiles $p\in \tau^{S}$ for some $S \subset \Z$ such that all dominoes of $p$ have matching colors on their sides. 
    A classical result of compactness implies that there exists a tiling of $\Z$ if and only if there exists valid patterns of any size (see for example \cite{LindMarcus}).
    This suggests the possibility of a general notion of tiling for any tensorial domino in terms of matrix powers.

    \begin{definition}[Tiling]
    A tensorial domino $T$ \emph{tiles the line} iff $\forall n\in \mathbb{N}, {T^n \neq 0}$. In other words, $T$ doesn't tile the line if and only if $T$ is nilpotent.
    \end{definition}

    This unified definition of tiling has subtly different interpretations depending on the type of tensorial domino you consider.
    In the possibilistic or probabilistic case, this directly corresponds to the existence of a valid Wang tiling by dominoes from the support.
    However, in the quantum case, matrices can have negative coefficients, so interference can come into play even if the support admits perfectly valid tilings. Here is an example of this phenomenon:

    \begin{example}\label{ex:1Dquantique}
    Consider the quantum domino $T=\frac{1}{2}\begin{pmatrix}
      1&1\\ -1& -1
    \end{pmatrix}$.
    Intuitively, it corresponds to the classical set of dominoes $\{ \includegraphics[scale=0.5]{white-white}, \includegraphics[scale=0.5]{black-white} \}$ “minus” $\{ \includegraphics[scale=0.5]{white-black}, \includegraphics[scale=0.5]{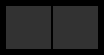} \}$.
    It has full support and hence its support has perfectly valid tilings of $\Z$ as a classical set of dominoes.
    However, interference negate the amplitude of any length $2$ (or larger) patterns as $T^2 = 0$, so the quantum domino does not tile the line.
    \end{example}

    This is our first example of quantum weirdness occurring in the formalism, and not the last.

  \subsection{(a)Periodicity}
    A tiling of the line is said to be \emph{periodic} if there exists a valid pattern starting and ending with the same color $c\in A$. A set of dominoes is said to be \emph{aperiodic} if it tiles the line but only in a non-periodic manner. 
    Given a possibilistic domino $T$, there exists a periodic pattern of size $n$ if and only if there exists $c\in A$ such that $\langle c|T^n |c\rangle \neq 0$.
    We can even count the number of periodic patterns of length $n$ with the formula $\sum_{c\in A} \langle c|T^n |c\rangle = \mathrm{tr}(T^n)$. This number is given by the trace of the matrix depicted as:
    \begin{center}
          \begin{tikzpicture}
	\begin{pgfonlayer}{nodelayer}
		\node [style=none] (498) at (0.5, 0.75) {};
		\node [style=none] (499) at (2, 0.75) {};
		\node [style=none] (500) at (2, -0.75) {};
		\node [style=none] (501) at (0.5, -0.75) {};
		\node [style=none] (502) at (1.25, 0) {$T$};
		\node [style=none] (508) at (0.5, 0) {};
		\node [style=none] (513) at (2.5, 0.75) {};
		\node [style=none] (514) at (4, 0.75) {};
		\node [style=none] (515) at (4, -0.75) {};
		\node [style=none] (516) at (2.5, -0.75) {};
		\node [style=none] (517) at (3.25, 0) {$T$};
		\node [style=none] (518) at (2, 0) {};
		\node [style=none] (519) at (2.5, 0) {};
		\node [style=none] (520) at (4, 0) {};
		\node [style=none] (521) at (4.5, 0) {};
		\node [style=none] (522) at (5.25, 0) {$...$};
		\node [style=none] (523) at (6.5, 0.75) {};
		\node [style=none] (524) at (8, 0.75) {};
		\node [style=none] (525) at (8, -0.75) {};
		\node [style=none] (526) at (6.5, -0.75) {};
		\node [style=none] (527) at (7.25, 0) {$T$};
		\node [style=none] (528) at (6, 0) {};
		\node [style=none] (529) at (6.5, 0) {};
		\node [style=none] (530) at (8.5, 0.75) {};
		\node [style=none] (531) at (10, 0.75) {};
		\node [style=none] (532) at (10, -0.75) {};
		\node [style=none] (533) at (8.5, -0.75) {};
		\node [style=none] (534) at (9.25, 0) {$T$};
		\node [style=none] (535) at (8, 0) {};
		\node [style=none] (536) at (8.5, 0) {};
		\node [style=none] (537) at (10, 0) {};
		\node [style=none] (538) at (10.5, 0) {};
		\node [style=none] (539) at (10.5, 0.75) {};
		\node [style=none] (540) at (12, 0.75) {};
		\node [style=none] (541) at (12, -0.75) {};
		\node [style=none] (542) at (10.5, -0.75) {};
		\node [style=none] (543) at (11.25, 0) {$T$};
		\node [style=none] (545) at (12, 0) {};
		\node [style=none] (547) at (0.5, -2) {};
		\node [style=none] (548) at (12, -2) {};
		\node [style=none] (549) at (13, -1) {};
		\node [style=none] (550) at (-0.5, -1) {};
	\end{pgfonlayer}
	\begin{pgfonlayer}{edgelayer}
		\draw (498.center) to (499.center);
		\draw (500.center) to (499.center);
		\draw (498.center) to (501.center);
		\draw (501.center) to (500.center);
		\draw (513.center) to (514.center);
		\draw (515.center) to (514.center);
		\draw (513.center) to (516.center);
		\draw (516.center) to (515.center);
		\draw (520.center) to (521.center);
		\draw (519.center) to (518.center);
		\draw (523.center) to (524.center);
		\draw (525.center) to (524.center);
		\draw (523.center) to (526.center);
		\draw (526.center) to (525.center);
		\draw (529.center) to (528.center);
		\draw (530.center) to (531.center);
		\draw (532.center) to (531.center);
		\draw (530.center) to (533.center);
		\draw (533.center) to (532.center);
		\draw (537.center) to (538.center);
		\draw (536.center) to (535.center);
		\draw (539.center) to (540.center);
		\draw (541.center) to (540.center);
		\draw (539.center) to (542.center);
		\draw (542.center) to (541.center);
		\draw [bend right=45] (508.center) to (550.center);
		\draw [bend right=45] (550.center) to (547.center);
		\draw (547.center) to (548.center);
		\draw [bend right=45] (548.center) to (549.center);
		\draw [bend right=45] (549.center) to (545.center);
	\end{pgfonlayer}
\end{tikzpicture}
    \end{center}

    Similarly, for a probabilistic domino, $\mathrm{tr}(T^n)$ is the probability of obtaining an admissible periodic length $n$ pattern by sampling $n$ random dominoes from the support.
    In both cases, we have that the set of dominoes is aperiodic if and only if $\forall c\in A,\forall n\in\N, \langle c|T^n |c\rangle = 0$, which is equivalent to $\forall n\in\N, \tr(T^n)=0$, as possibilistic and probabilistic dominoes have non-negative coefficients.

    The situation is more subtle for quantum dominoes. Then the trace is the sum of the amplitudes of having a periodic pattern with each possible end. Since those situations are mutually exclusive possibilities it follows from the usual rules of quantum mechanics that their sum, that is the trace, is the amplitudes for obtaining a periodic pattern from the interaction of $n$ dominoes.
    However, the previous equivalence does not hold, and we end up with two non-equivalent definitions of aperiodicity in the general case. In this paper we will use the definition of aperiodicity in terms of trace, as it is the most relevant from a quantum point of view.

    \begin{definition}[Trace aperiodicity]
    A tensorial domino $T$ is said to be \emph{trace-aperiodic} if and only if $\forall n\in \mathbb{N}, \mathrm{tr}(T^n) = 0$. Else it is said \emph{trace-periodic}.
    \end{definition}


    Tileability (equivalently, the nilpotency of the matrix) can be reformulated as follow.

    \begin{proposition}[Trace characterization of nilpotency]\label{trace}
    A $n\times n$ matrix $T$ is nilpotent if and only if forall $1\leq k\leq n, \mathrm{tr}(T^k)=0$.
    \end{proposition}

    We can find this characterization in many linear algebra textbooks, for example, \cite{lang2012algebra}. 
    If we interpret it in terms of possibilistic dominoes, it is precisely the classical result that an aperiodic set of dominoes do not tile the line.
    The proposition also applies in $\R$ or $\C$, allowing us to generalize this result in the general tensorial case:

    \begin{theorem}
    A tensorial domino is trace-aperiodic if and only if it does not tile the line.
    \end{theorem}

    Even in the possibilistic case, it is interesting to link this demonstration through nilpotency to the more common one relying on graphs (see for example \cite{LindMarcus}).
    It follows directly by interpreting the tensorial domino as an adjacency matrix. The coefficients of $T^n$ correspond to the number of length $n$ paths in the graph, and then the absence of any infinite path exactly corresponds to nilpotency. In the same idea, $\mathrm{tr}(T^n)$ counts the number of size $n$ cycles in the graph. Thus, the result states that an acyclic graph has no infinite path.

    However, the strength of the nilpotency approach is that it still holds for any complex matrix. So the result still holds for quantum dominoes, even with our generalized notions of tiling and periodic patterns. This is not obvious, as one could have expected cases when interference can suppress periodic patterns while still allowing arbitrary large configurations. Or, on the contrary, situations when we have periodic patterns but interferences prevent large configurations. Here is a typical example:

    \begin{example}
    The quantum domino $\frac{1}{\sqrt{2}}\begin{pmatrix}
      1&0\\0&i
    \end{pmatrix}$ has size one periodic patterns as $\mathrm{tr}(T)= \frac{1+i}{\sqrt{2}}$. However, it has no periodic patterns of size two as $\mathrm{tr}(T^2)=\frac{1}{2}\mathrm{tr}\left(\begin{pmatrix}
      1&0\\0&-1
    \end{pmatrix}\right)=0$.
    \end{example}

    So in the quantum setting, we have to drop the classical intuition that a domino with size $k$ periodic patterns will admit size $kn$ periodic patterns for all $n$. 
    Hopefully, we still have the following:

    \begin{lemma}\label{lem:perpav}
    If $T$ is non-aperiodic,  then it admits arbitrarily large periodic patterns. In other words, for any $N\in \mathbb{N}$, there exists $k\geq N$ such that $\mathrm{T^k}\neq 0$.
    \end{lemma}

    \begin{proof}
      Let $k$ be such that $\mathrm{T^k}\neq 0 $ then it follows that $T^k $ is not nilpotent and then $T^{kl}$ is not nilpotent for any $l\in \mathbb{N}$. So by \cref{trace}, there is an integer $1 \leq m\leq n $ such that $\mathrm{T^{klm}}\neq 0$. 
    \end{proof}

    Happily, the link between periodicity and tilability still holds. Quantum dominoes are no wilder than classical ones on this aspect. We will see that it is no longer true for quantum Wang tiles.

    For the classical case, trace-aperiodicity is equivalent to the aperiodicity of the support.
    
    \begin{proposition}
    Let $T$ be a possibilistic or probabilistic domino. Then it is trace-aperiodic if and only $\mathrm{supp}(T)$ is aperiodic.
    \end{proposition}

    This is no longer true in the quantum case, when aperiodicity of the support is stronger, as empathized by the following proposition.
    
    \begin{proposition}\label{prop:aptotap}
    Let $T$ be a quantum domino. If its support is aperiodic, then $T$ is trace-aperiodic.
    \end{proposition}
    
    \begin{proof}
      Let $T$ be such that $\mathrm{supp}(T)$ is aperiodic. Let $U$ be the possibilistic domino associated with $\mathrm{supp}(T)$ and take $n\in \N$.
    Then by assumption, for any $a\in A, $ \[ \bra{a} U^n \ket{a} = 0 . \]
    Therefore,  \[ \bra{a} \sum_{a_1 \cdots a_n\in A^n} \prod_{i=1}^n  U^{a_i}_{a_{i+1}} \ket{a} = 0  \]
    with $a_0 = a_{n+1} = a$.
    As $U$ has only non-negative coefficients, all the terms of the sum are zero, meaning that for each word $a_1 \cdots a_n\in A^n$, there is one $i\in \{1,\hdots, n\}$ such that $U^{a_i}_{a_{i+1}}=0$.
    By definition of $U$, we will also have $T^{a_i}_{a_{i+1}0}=0$, and so
    \[ \bra{a} \sum_{a_1 \cdots a_n\in A^n} \prod_{i=1}^n  T^{a_i}_{a_{i+1}} \ket{a} = 0  \]
    Implying that $\tr(T^n)=0$, therefore $T$ is trace-aperiodic.
    \end{proof}
	
	A counter example to the converse is given by \cref{ex:1Dquantique}. Intuitively, the aperiodicity of the support means that the tileset is just a ``quantum version'' of a classical aperiodic tileset.
   



    Overall, we have the following implications in the 1-dimensional quantum case, with the addition of the right-left implication in the probabilistic and possibilistic case.

    \begin{center}
      $\mathrm{supp}(T)$ doesn't tile $\Z~\Leftrightarrow~\mathrm{supp}(T)$ aperiodic $~\Rightarrow~$ $T$ trace-aperiodic $~\Leftrightarrow~ T$ doesn't tile $\Z$
    \end{center}


	The fact that aperiodicity of the support implies trace aperiodicity and that trace aperiodicity is equivalent to non-tileability, suggest that trace aperidicity is the right extension of the concept of aperiodicity to the quantum case.
    Now that we have extensively studied the one-dimensional case, we will extend all definitions and interpretations to the two-dimensional one. Sadly, we will not be able to rely on familiar matrix algebra anymore, and the full power of the tensor formalism will be required there.

\section{Tensorial tiles} \label{sec:2D}
  \subsection{Definition}
    \begin{wrapfigure}{r}{0.35\textwidth}
      \centering
      \includegraphics[scale=1]{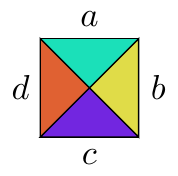}\quad
      \includegraphics[scale=0.5]{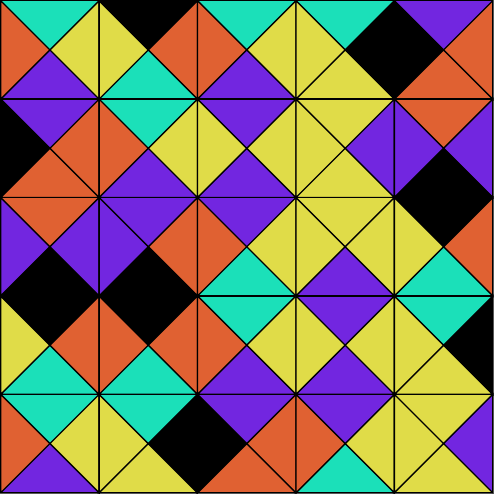}
      \caption{Example of Wang tile and portion of $\Z^2$ tiled.}
    \end{wrapfigure}
    
    As before, let $A$ be a finite alphabet. A (two-dimensional) \emph{Wang tile} is a quadruplet $(a,b,c,d)\in A^4$, each color corresponding to a side of the tile. Let $\tau$ be a finite set of Wang tiles.
    A pattern $p\in \tau^S$ of finite support $S\subset \Z^2$ is said to be \emph{valid} for $\tau$ if the color of the sides of each tile of the pattern match, i.e. for all $(i,j)\in D,$ write $x_{(i,j)} = (a,b,c,d)$, then
    \[
          \begin{cases}
            (i+1, j) \in D \Rightarrow \exists f,g,h,& x_{(i+1, j)} = (f,g,h,b) \\
            (i-1, j) \in D \Rightarrow \exists f,g,h,& x_{(i-1, j)} = (f,d,g,h) \\
            (i, j+1) \in D \Rightarrow \exists f,g,h,& x_{(i, j+1)} = (f,g,a,h) \\
            (i, j-1) \in D \Rightarrow \exists f,g,h,& x_{(i, j-1)} = (c,f,g,h) \\
          \end{cases}
    \]
        We say that $\tau$ \emph{tiles the plane $\Z^2$} if the above holds for $S = \Z^2$.
        
    Like for the case of $\Z$, compactness of the space $\A^{\Z^2}$ implies that $\tau$ tiles the plane if and only if there exists arbitrarily large valid patterns.
    A valid pattern (or a tiling) $x$ of support $S$ is said to be \emph{weakly periodic} of period $\vec u\neq 0$ if for all $\vec v\in S$, $\forall k\in \Z, \vec v+k\vec u\in S \Rightarrow x_{\vec v+k\vec u} = x_{\vec v}$. It is \emph{strongly periodic} if it has two non-colinear vectors of periodicity. 
    It turns out that a tileset admits a weakly periodic tiling if and only if it admits a strongly periodic one, thus we simply say that $\tau$ is \emph{aperiodic} if it has valid tilings but no weakly (equivalently no strongly) periodic ones.
    A fundamental result in tilings is that in dimension two, unlike dimension one, there exist aperiodic tilesets~\cite{Berger}. We aim to investigate if quantum interference allow us to build new \emph{kinds} of aperiodic tilesets.

    \medskip

    Our definition of tensorial dominoes on the line directly extends into the definition of tensorial tiles in the plane.

    \begin{definition}
    A \textbf{tensorial tile} is a 4-tensor, that can be seen as $T\in \mathbb{C}^{|A|} \otimes \mathbb{C}^{|A|} \otimes \mathbb{C}^{|A|} \otimes \mathbb{C}^{|A|}$.
    It is:
    \begin{itemize}
      \item \textbf{possibilistic} if and only if $\forall (x,y,z,t)\in A^4, T_{x,y,z,t}\in \{0,1\}$.
      \item \textbf{probabilistic} if and only if $\forall (x,y,z,t)\in A^4, T_{x,y,z,t}\in [0,1]$ and $\sum_{(x,y,z,t)\in A^4} T_{x,y,z,t}=1$.
      \item \textbf{quantum} if and only if $\sum_{(x,y,z,t)\in A^4} |T_{x,y,z,t}|^2 =1$.
    \end{itemize}
    \end{definition}

    By convention, a tensorial tile will be depicted as:
    \begin{tikzpicture}
	\begin{pgfonlayer}{nodelayer}
		\node [style=none] (498) at (2.25, 0.75) {};
		\node [style=none] (499) at (3.75, 0.75) {};
		\node [style=none] (500) at (3.75, -0.75) {};
		\node [style=none] (501) at (2.25, -0.75) {};
		\node [style=none] (502) at (3, 0) {$T$};
		\node [style=none] (507) at (2, 0) {};
		\node [style=none] (508) at (2.25, 0) {};
		\node [style=none] (509) at (3, 0.75) {};
		\node [style=none] (510) at (3, 1) {};
		\node [style=none] (511) at (3.75, 0) {};
		\node [style=none] (512) at (4, 0) {};
		\node [style=none] (513) at (3, -0.75) {};
		\node [style=none] (514) at (3, -1) {};
		\node [style=none] (575) at (-1, 0) {$T_{x,y,z,t}$};
		\node [style=none] (577) at (3, 1.5) {$x$};
		\node [style=none] (578) at (4.5, 0) {$y$};
		\node [style=none] (579) at (3, -1.5) {$z$};
		\node [style=none] (580) at (1.5, 0) {$t$};
		\node [style=none] (581) at (0.5, 0) {$=$};
	\end{pgfonlayer}
	\begin{pgfonlayer}{edgelayer}
		\draw (498.center) to (499.center);
		\draw (500.center) to (499.center);
		\draw (498.center) to (501.center);
		\draw (501.center) to (500.center);
		\draw (510.center) to (509.center);
		\draw (511.center) to (512.center);
		\draw (513.center) to (514.center);
		\draw (508.center) to (507.center);
	\end{pgfonlayer}
\end{tikzpicture}

    Notice that we are writing indices in clockwise order, starting by the upper one. The interpretation of the coefficients for the different types of tensorial tile are completely analogous to tensorial dominoes. However, the situation is more intricate when it comes to composing tiles.

    We also define product and union of tensorial tile in a way that we will again have: $ \mathrm{supp}(T\times L)=\mathrm{supp}(T) \times \mathrm{supp}(L)$ and $ \mathrm{supp}(T\uplus L)= \mathrm{supp}(T)\uplus \mathrm{supp}(L)$.

    \begin{definition}[Union]
    The \textbf{union} of two tensorial tiles $T\in \mathbb{C}^{d^{4}}$ and $L\in \mathbb{C}^{k^{4}}$ is the tensorial tile $T\uplus L \in \mathbb{C}^{(d+k)^{4}}$ whose coefficients are:

    \begin{center}
      $(T\uplus L)_{x,y,z,t} = \begin{cases}
        T_{x,y,z,t} \text{ if } 1\leq x,y,z,t \leq d \\
        L_{x-d,y-d,z-d,t-d} \text{ if } d < x,y,z,t \leq d+k \\
        0 \text{ otherwise.}
      \end{cases}$
    \end{center}
    which corresponds to the direct sum of the tensors.

    \end{definition}

    The product corresponds to the tensor product with an additional permutation of indices.

    \begin{definition}[Product]
    The \textbf{product} of two tensorial tiles $T\in \mathbb{C}^{d^{4}}$ and $T'\in \mathbb{C}^{d'^{4}}$, $(T\times T')\in \mathbb{C}^{(dd')^{4}}$ is defined as: $(T\times T')_{x,a,y,b,c,z,d,t}= T_{x,y,z,t}T'_{a,b,c,d}$, graphically:
    \begin{center}
                  \scalebox{0.75}{
                    \begin{tikzpicture}
	\begin{pgfonlayer}{nodelayer}
		\node [style=none] (498) at (0.75, 1.75) {};
		\node [style=none] (499) at (2.25, 1.75) {};
		\node [style=none] (500) at (2.25, 0.25) {};
		\node [style=none] (501) at (0.75, 0.25) {};
		\node [style=none] (502) at (1.5, 1) {$T$};
		\node [style=none] (507) at (0, 1) {};
		\node [style=none] (508) at (0.75, 1) {};
		\node [style=none] (509) at (1.5, 1.75) {};
		\node [style=none] (510) at (1.5, 2.5) {};
		\node [style=none] (511) at (2.25, 1) {};
		\node [style=none] (512) at (5, 1) {};
		\node [style=none] (513) at (1.5, 0.25) {};
		\node [style=none] (514) at (1.5, -2.5) {};
		\node [style=none] (581) at (2.75, -0.25) {};
		\node [style=none] (582) at (4.25, -0.25) {};
		\node [style=none] (583) at (4.25, -1.75) {};
		\node [style=none] (584) at (2.75, -1.75) {};
		\node [style=none] (585) at (3.5, -1) {$T'$};
		\node [style=none] (586) at (0, -1) {};
		\node [style=none] (587) at (2.75, -1) {};
		\node [style=none] (588) at (3.5, -0.25) {};
		\node [style=none] (589) at (3.5, 2.5) {};
		\node [style=none] (590) at (4.25, -1) {};
		\node [style=none] (591) at (5, -1) {};
		\node [style=none] (592) at (3.5, -1.75) {};
		\node [style=none] (593) at (3.5, -2.5) {};
	\end{pgfonlayer}
	\begin{pgfonlayer}{edgelayer}
		\draw (498.center) to (499.center);
		\draw (500.center) to (499.center);
		\draw (498.center) to (501.center);
		\draw (501.center) to (500.center);
		\draw (510.center) to (509.center);
		\draw (511.center) to (512.center);
		\draw (513.center) to (514.center);
		\draw (508.center) to (507.center);
		\draw (581.center) to (582.center);
		\draw (583.center) to (582.center);
		\draw (581.center) to (584.center);
		\draw (584.center) to (583.center);
		\draw (589.center) to (588.center);
		\draw (590.center) to (591.center);
		\draw (592.center) to (593.center);
		\draw (587.center) to (586.center);
	\end{pgfonlayer}
\end{tikzpicture}
                  } 
    \end{center}
    \end{definition}

    \newpage 
    \subsection{Tilings as tensor networks}
    \begin{wrapfigure}{r}{0.4\textwidth}
      \centering
      \vspace*{-2em}
      \scalebox{0.9}{\input{zindex.tikz}}
    \end{wrapfigure}
    In two dimensions we can't directly see a tensorial tile as a matrix, so now juxtapositions of tiles will not be represented as matrix product but as tensor contraction. The idea is to place some tensorial tiles in the plane and perform a contraction each time two tiles are adjacent. Thus given a tensorial tile $T$ and a finite part of the plane $S$ we should be able to construct a tensor $S\cdot T$ that represents $S$ covered with the tensor $T$ on each of its positions, and contracted where necessary.
          We start by fixing a set of abstract indices themselves indexed by couples of half integers $\left(i_{a,b}\right)_{(a,b)\in \mathbb{Z}\times (\mathbb{Z}+\frac{1}{2}) \bigcup (\mathbb{Z}+\frac{1}{2}) ^\times \mathbb{Z}}$, they corresponds to the red dots on the picture, one for each possible link.
          
    This allows us to match $S$ with a tensor in abstract index notation.

    \begin{definition}
      Given a tensorial tile $T$ and finite \emph{shape} $S \subset \Z^2$, we define: \[ S\cdot T = \prod_{(x,y)\in S} T_{i_{(x,y+\frac{1}{2})},i_{(x+\frac{1}{2},y)},i_{(x,y-\frac{1}{2})},i_{(x-\frac{1}{2},y)}} \]
      
    \end{definition}

    Thanks to Einstein's summation convention, each repeated indices are summed, and the necessary contractions occur. Of course, once the tensor have been properly defined, we are free to rename the indices, thus the final result is translation invariant: $S\cdot T= t(S)\cdot T$ for any translation $t$.

    \begin{example}
      We consider the shape $S=\{(-1,1),(1,1),(0,0),(1,0),(0,-1)\}$. Then the corresponding tensor $S\cdot T$ is: \begin{align*}
        T_{i_{(-1,\frac{3}{2})},i_{(-\frac{1}{2},1)},i_{(-1,\frac{1}{2})},i_{(-\frac{3}{2},1)}}T_{i_{(1,\frac{3}{2})},i_{(\frac{3}{2},1)},i_{(1,\frac{1}{2})},i_{(\frac{1}{2},1)}}T_{i_{(0,\frac{1}{2})},i_{(\frac{1}{2},0)},i_{(0,-\frac{1}{2})},i_{(-\frac{1}{2},0)}} \\
        T_{i_{(1,\frac{1}{2})},i_{(\frac{3}{2},0)},i_{(1,-\frac{1}{2})},i_{(\frac{1}{2},0)}}T_{i_{(0,-\frac{1}{2})},i_{(\frac{1}{2},-1)},i_{(0,-\frac{3}{2})},i_{(-\frac{1}{2},-1)}}
      \end{align*}
      Which after renaming of the indices is: $T_{a,b,c,d}T_{e,f,g,h}T_{i,j,k,l}T_{g,m,n,j}T_{k,o,p,q}$, a $14$-tensor with $14$ indices and two repeated ones, $g$ and $j$. Graphically:
      
      \begin{center}
        \begin{tikzpicture}
	\begin{pgfonlayer}{nodelayer}
		\node [style=none] (496) at (0, 0) {$S$};
		\node [style=none] (497) at (1, 0) {$=$};
		\node [style=none] (498) at (13.25, 2) {};
		\node [style=none] (499) at (14.25, 2) {};
		\node [style=none] (500) at (14.25, 1) {};
		\node [style=none] (501) at (13.25, 1) {};
		\node [style=none] (502) at (13.75, 1.5) {$T$};
		\node [style=none] (507) at (13, 1.5) {};
		\node [style=none] (508) at (13.25, 1.5) {};
		\node [style=none] (509) at (13.75, 2) {};
		\node [style=none] (510) at (13.75, 2.25) {};
		\node [style=none] (511) at (14.25, 1.5) {};
		\node [style=none] (512) at (14.5, 1.5) {};
		\node [style=none] (513) at (13.75, 1) {};
		\node [style=none] (514) at (13.75, 0.75) {};
		\node [style=none] (515) at (16.25, 0.5) {};
		\node [style=none] (516) at (17.25, 0.5) {};
		\node [style=none] (517) at (17.25, -0.5) {};
		\node [style=none] (518) at (16.25, -0.5) {};
		\node [style=none] (519) at (16.75, 0) {$T$};
		\node [style=none] (520) at (16.25, 0) {};
		\node [style=none] (521) at (16.75, 0.5) {};
		\node [style=none] (522) at (16.75, -0.5) {};
		\node [style=none] (523) at (16.75, -0.75) {};
		\node [style=none] (524) at (14.75, 0.5) {};
		\node [style=none] (525) at (15.75, 0.5) {};
		\node [style=none] (526) at (15.75, -0.5) {};
		\node [style=none] (527) at (14.75, -0.5) {};
		\node [style=none] (528) at (15.25, 0) {$T$};
		\node [style=none] (529) at (15.25, 0.5) {};
		\node [style=none] (530) at (15.25, 0.75) {};
		\node [style=none] (531) at (15.75, 0) {};
		\node [style=none] (532) at (15.25, -0.5) {};
		\node [style=none] (533) at (16.25, 2) {};
		\node [style=none] (534) at (17.25, 2) {};
		\node [style=none] (535) at (17.25, 1) {};
		\node [style=none] (536) at (16.25, 1) {};
		\node [style=none] (537) at (16.75, 1.5) {$T$};
		\node [style=none] (538) at (16.75, 2) {};
		\node [style=none] (539) at (16.75, 2.25) {};
		\node [style=none] (540) at (16.75, 1) {};
		\node [style=none] (553) at (14.75, -1) {};
		\node [style=none] (554) at (15.75, -1) {};
		\node [style=none] (555) at (15.75, -2) {};
		\node [style=none] (556) at (14.75, -2) {};
		\node [style=none] (557) at (15.25, -1.5) {$T$};
		\node [style=none] (558) at (15.75, -1.5) {};
		\node [style=none] (559) at (16, -1.5) {};
		\node [style=none] (560) at (15.25, -1) {};
		\node [style=none] (561) at (14.5, -1.5) {};
		\node [style=none] (562) at (14.75, -1.5) {};
		\node [style=none] (563) at (15.25, -2) {};
		\node [style=none] (564) at (15.25, -2.25) {};
		\node [style=none] (565) at (16.25, 1.5) {};
		\node [style=none] (566) at (16, 1.5) {};
		\node [style=none] (569) at (14.75, 0) {};
		\node [style=none] (570) at (14.5, 0) {};
		\node [style=none] (571) at (17.25, 1.5) {};
		\node [style=none] (572) at (17.5, 1.5) {};
		\node [style=none] (573) at (17.25, 0) {};
		\node [style=none] (574) at (17.5, 0) {};
		\node [style=none] (575) at (10.5, 0) {$S\cdot T$};
		\node [style=none] (576) at (12, 0) {$=$};
		\node [style=none] (577) at (2, 2.25) {};
		\node [style=none] (578) at (3.5, 2.25) {};
		\node [style=none] (579) at (3.5, 0.75) {};
		\node [style=none] (580) at (6.5, 2.25) {};
		\node [style=none] (581) at (5, 0.75) {};
		\node [style=none] (582) at (5, 2.25) {};
		\node [style=none] (583) at (6.5, -0.75) {};
		\node [style=none] (584) at (5, -0.75) {};
		\node [style=none] (585) at (5, -2.25) {};
		\node [style=none] (586) at (3.5, -2.25) {};
		\node [style=none] (587) at (3.5, -0.75) {};
		\node [style=none] (588) at (2, 0.75) {};
	\end{pgfonlayer}
	\begin{pgfonlayer}{edgelayer}
		\draw (498.center) to (499.center);
		\draw (500.center) to (499.center);
		\draw (498.center) to (501.center);
		\draw (501.center) to (500.center);
		\draw (510.center) to (509.center);
		\draw (511.center) to (512.center);
		\draw (513.center) to (514.center);
		\draw (508.center) to (507.center);
		\draw (515.center) to (516.center);
		\draw (517.center) to (516.center);
		\draw (515.center) to (518.center);
		\draw (518.center) to (517.center);
		\draw (522.center) to (523.center);
		\draw (524.center) to (525.center);
		\draw (526.center) to (525.center);
		\draw (524.center) to (527.center);
		\draw (527.center) to (526.center);
		\draw [style=pointille] (530.center) to (529.center);
		\draw (533.center) to (534.center);
		\draw (535.center) to (534.center);
		\draw (533.center) to (536.center);
		\draw (536.center) to (535.center);
		\draw (539.center) to (538.center);
		\draw (553.center) to (554.center);
		\draw (555.center) to (554.center);
		\draw (553.center) to (556.center);
		\draw (556.center) to (555.center);
		\draw (558.center) to (559.center);
		\draw (530.center) to (529.center);
		\draw (520.center) to (531.center);
		\draw (540.center) to (521.center);
		\draw (532.center) to (560.center);
		\draw (561.center) to (562.center);
		\draw (563.center) to (564.center);
		\draw (565.center) to (566.center);
		\draw (569.center) to (570.center);
		\draw (571.center) to (572.center);
		\draw (573.center) to (574.center);
		\draw [style=grey] (579.center)
			 to (587.center)
			 to (586.center)
			 to (585.center)
			 to (584.center)
			 to (583.center)
			 to (580.center)
			 to (582.center)
			 to (581.center)
			 to cycle
			 to (578.center)
			 to (577.center)
			 to (588.center)
			 to cycle;
	\end{pgfonlayer}
\end{tikzpicture}
      \end{center}
      
      Where we can see the two links corresponding to indices $g$ and $j$.
    \end{example}

    When taking the union (i.e. direct sum) of tensorial tiles, linearity gives us the following property.

    \begin{proposition}
            \label{prop:sum}
      Let $T$, $L$ be two tensorial tiles and $S$ a shape. Then 
      $ S\cdot (T \uplus L) = (S\cdot T) \uplus (S \cdot L) $.
    \end{proposition}

    \begin{proof}
      Unfolding the definitions, we have:
      
      \begin{align*}
        S\cdot (T \uplus L)&=\prod_{(x,y)\in S} (T \uplus L)_{i_{(x,y+\frac{1}{2})},i_{(x+\frac{1}{2},y)},i_{(x,y-\frac{1}{2})},i_{(x-\frac{1}{2},y)}}\\
        &= \left(\prod_{(x,y)\in S} T_{i_{(x,y+\frac{1}{2})},i_{(x+\frac{1}{2},y)},i_{(x,y-\frac{1}{2})},i_{(x-\frac{1}{2},y)}}  \right)\\
        &\uplus \left(\prod_{(x,y)\in S} L_{i_{(x,y+\frac{1}{2})},i_{(x+\frac{1}{2},y)},i_{(x,y-\frac{1}{2})},i_{(x-\frac{1}{2},y)}}  \right)\\
        &=(S\cdot T) \uplus (S \cdot L)
      \end{align*}
    \end{proof}

    For a possibilistic tile, the coefficients of $S\cdot T$ correspond to the number of possible ways to tile the shape $S$ while satisfying the border conditions given by the indices. As for tensorial dominoes, the coefficients of $S\cdot T$ for probabilistic and quantum tiles correspond respectively to probabilities and amplitudes.
    For possibilistic dominoes, the existence of a valid (infinite) tiling is equivalent to the fact that any (finite) shape should have a valid tiling.
    This suggests the following generalized notion of tiling:

    \begin{definition}
      A tensorial tile $T$ tiles the plane if and only if for all finite $S\subset \mathbb{Z}^2, S\cdot T\neq 0$.
    \end{definition}

    It is easy to see that any sub-shape of a non-zero tensor will have a non-zero tensor.
    \begin{lemma}\label{lemma:subshape}
      Let $T$ be a tensorial tile and $S'\subseteq S \subset \Z^2 $ such that $S\cdot T\neq 0$. Then $S'\cdot T\neq 0$.
    \end{lemma}

    \noindent
    Therefore, it is enough to check that a tileset tiles infinitely many rectangles to know if it tiles the whole plane.
    We will write $R_{m,n} = \llbracket 1,m \rrbracket \times \llbracket 1,n \rrbracket \subset \Z^2$.

    \begin{proposition}
      A tensorial tile $T$ tiles the plane if and only if for all $m,n$, ${R_{m,n}\cdot T\neq 0}$.
    \end{proposition}



    \begin{example}
      A very small yet interesting classical tileset is the following one:
      \begin{center}
          \begin{tikzpicture}[scale=1.5]
               \wang{0}{0}{white}{white}{white}{black!80}
               \wang{1.5}{0}{white}{white}{black!80}{white}
               \wang{3}{0}{white}{black!80}{white}{white}
               \wang{4.5}{0}{black!80}{white}{white}{white}
          \end{tikzpicture}
      \end{center}
      It is the famous dimer model, i.e. tilings of the grids by $2\times 1$ rectangles, which has been studied extensively \cite{Kasteleyn, TemperleyFisher, PlanarDimers}.
      As a possibilistic tile, it is represented by the following binary tensor, with $x,y,z,t\in\{0,1\}$:
      \[ T_{xyzt} = 1 \text{ if and only if } x+y+z+t=1 .\]
      It turns out that this tensor corresponds to the \emph{black spider} of the ZW-calculus, a well studied graphical language for quantum computing.
      In \cite{ZW}, this correspondence allowed us to make a new combinatorial interpretation of the ZW calculus. And the study of its fragment representing dimer tilings allowed us to develop new techniques to count dimer tilings based on diagram rewriting.
    \end{example}

    \begin{example}
      One can derive a quantum tileset from the previous dimer model, for example by considering a model with two different complex weights on the horizontal and vertical dimers. For that we define $T$ by:
      \[
      \begin{cases}
        T_{1,0,0,0} = \frac{\sqrt{7}}{4}\\
        T_{0,0,1,0} = \frac{\sqrt{7}}{4}\\
        T_{0,1,0,0} = -\frac{1}{4}\\
        T_{0,0,0,1} = \frac{1}{4}\\
        T_{x,y,z,t} = 0 & \text{otherwise}
      \end{cases}
      \]
      Then, it is a quantum tile, and it represents the dimer tileset with amplitudes $\frac{\sqrt{7}}{4}$ on tiles \begin{tikzpicture}[scale=1]
               \wang{0}{-0.5}{black!80}{white}{white}{white}
      \end{tikzpicture} 
      and 
      \begin{tikzpicture}[scale=1]
               \wang{0}{-0.5}{white}{white}{black!80}{white}
      \end{tikzpicture}, amplitude $-\frac{1}{4}$ on
      \begin{tikzpicture}[scale=1]
               \wang{0}{-0.5}{white}{black!80}{white}{white}
      \end{tikzpicture} 
      and $\frac{1}{4}$ on
      \begin{tikzpicture}[scale=1]
               \wang{0}{-0.5}{white}{white}{white}{black!80}
      \end{tikzpicture}.
      In other words, vertical dimers have weight $\frac{7}{16}$ and horizontal ones $-\frac{1}{16}$.
    \end{example}

  \subsection{(a)Periodicities}
    As our definition of tensorial tiles allows us to only talk about sets of valid patterns, and to stay as close as possible to the usual terminology, we will only define notions of aperiodicity, not periodicity.

    In the general tensorial case, it is not immediate that weakly and strongly aperiodic are equivalent. We will see in this section that they are, just like in the classical case.

    First, let us define a directional trace along some vector.




    \begin{definition}[Directional trace]
      Let $T$ be a tensorial tile, $\vec u\in\Z^2-\{(0,0)\}$ and $S\subset \Z^2$. The \emph{trace along $\vec u = (a,b)$} of $S\cdot T$ is:
      \[ \tr_{\vec u}(S\cdot T) =  \prod_{(x,y)\in S} T_{i_{(x,y+\frac{1}{2})\hspace{-0.55em}\mod \vec u},~i_{(x+\frac{1}{2},y)\hspace{-0.55em}\mod \vec u},~i_{(x,y-\frac{1}{2})\hspace{-0.55em}\mod \vec u},~i_{(x-\frac{1}{2},y)\hspace{-0.55em}\mod \vec u}} \]
    \end{definition}

    Intuitively, the directional trace will sum together all the indices that can be ``linked'' by the given vector, the other borders being untouched and their indices staying free.
    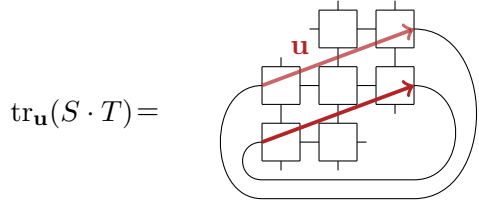
\begin{figure}[ht!]
      \centering
      \begin{tikzpicture}
	\begin{pgfonlayer}{nodelayer}
		\node [style=none] (515) at (7.5, -1.25) {};
		\node [style=none] (516) at (7.5, -0.25) {};
		\node [style=none] (517) at (6.5, -0.25) {};
		\node [style=none] (518) at (6.5, -1.25) {};
		\node [style=none] (520) at (7, -1.25) {};
		\node [style=none] (521) at (7.5, -0.75) {};
		\node [style=none] (522) at (6.5, -0.75) {};
		\node [style=none] (523) at (6, -0.75) {};
		\node [style=none] (531) at (7, -1.5) {};
		\node [style=none] (540) at (7.75, -0.75) {};
		\node [style=none] (573) at (7, -0.25) {};
		\node [style=none] (574) at (7, 0.25) {};
		\node [style=none] (575) at (0, 0) {$\mathrm{tr}_{\vec u}(S\cdot T)$};
		\node [style=none] (576) at (2, 0) {$=$};
		\node [style=none] (604) at (6, -1.25) {};
		\node [style=none] (605) at (6, -0.25) {};
		\node [style=none] (606) at (5, -0.25) {};
		\node [style=none] (607) at (5, -1.25) {};
		\node [style=none] (609) at (5.5, -0.25) {};
		\node [style=none] (610) at (5.5, 0.25) {};
		\node [style=none] (615) at (7.5, 0.25) {};
		\node [style=none] (616) at (7.5, 1.25) {};
		\node [style=none] (617) at (6.5, 1.25) {};
		\node [style=none] (618) at (6.5, 0.25) {};
		\node [style=none] (621) at (7.5, 0.75) {};
		\node [style=none] (622) at (6.5, 0.75) {};
		\node [style=none] (623) at (6, 0.75) {};
		\node [style=none] (631) at (8, 0.75) {};
		\node [style=none] (636) at (7, 1.25) {};
		\node [style=none] (637) at (7, 1.75) {};
		\node [style=none] (638) at (6, 0.25) {};
		\node [style=none] (639) at (6, 1.25) {};
		\node [style=none] (640) at (5, 1.25) {};
		\node [style=none] (641) at (5, 0.25) {};
		\node [style=none] (645) at (5, 0.75) {};
		\node [style=none] (646) at (5, -2.25) {};
		\node [style=none] (650) at (7.5, 1.75) {};
		\node [style=none] (651) at (7.5, 2.75) {};
		\node [style=none] (652) at (6.5, 2.75) {};
		\node [style=none] (653) at (6.5, 1.75) {};
		\node [style=none] (655) at (7.5, 2.25) {};
		\node [style=none] (665) at (8, 2.25) {};
		\node [style=none] (668) at (7, 2.75) {};
		\node [style=none] (669) at (7, 3) {};
		\node [style=none] (688) at (8.5, 0) {};
		\node [style=none] (689) at (8.5, 0.25) {};
		\node [style=none] (690) at (9, 0.25) {};
		\node [style=none] (691) at (9, 1.25) {};
		\node [style=none] (692) at (8, 1.25) {};
		\node [style=none] (693) at (8, 0.25) {};
		\node [style=none] (697) at (8.5, 1.75) {};
		\node [style=none] (698) at (8.5, 1.25) {};
		\node [style=none] (701) at (5.5, -1.25) {};
		\node [style=none] (702) at (5.5, -1.5) {};
		\node [style=none] (703) at (5.5, 1.5) {};
		\node [style=none] (704) at (5.5, 1.25) {};
		\node [style=none] (706) at (6, 2.25) {};
		\node [style=none] (736) at (9, 1.75) {};
		\node [style=none] (737) at (9, 2.75) {};
		\node [style=none] (738) at (8, 2.75) {};
		\node [style=none] (739) at (8, 1.75) {};
		\node [style=none] (743) at (8.5, 3) {};
		\node [style=none] (744) at (8.5, 2.75) {};
		\node [style=none] (747) at (6.25, 2.25) {};
		\node [style=none] (748) at (6.5, 2.25) {};
		\node [style=none] (751) at (5, -0.75) {};
		\node [style=none] (752) at (9, 0.75) {};
		\node [style=red text] (754) at (6, 1.75) {$\vec u$};
		\node [style=none] (756) at (5, -1.75) {};
		\node [style=none] (758) at (9, -1.75) {};
		\node [style=none] (759) at (9, 2.25) {};
		\node [style=none] (760) at (9, -2.25) {};
	\end{pgfonlayer}
	\begin{pgfonlayer}{edgelayer}
		\draw (515.center) to (516.center);
		\draw (517.center) to (516.center);
		\draw (515.center) to (518.center);
		\draw (518.center) to (517.center);
		\draw (522.center) to (523.center);
		\draw (520.center) to (531.center);
		\draw (540.center) to (521.center);
		\draw (573.center) to (574.center);
		\draw (604.center) to (605.center);
		\draw (606.center) to (605.center);
		\draw (604.center) to (607.center);
		\draw (607.center) to (606.center);
		\draw (609.center) to (610.center);
		\draw (615.center) to (616.center);
		\draw (617.center) to (616.center);
		\draw (615.center) to (618.center);
		\draw (618.center) to (617.center);
		\draw (622.center) to (623.center);
		\draw (631.center) to (621.center);
		\draw (636.center) to (637.center);
		\draw (638.center) to (639.center);
		\draw (640.center) to (639.center);
		\draw (638.center) to (641.center);
		\draw (641.center) to (640.center);
		\draw [bend right=90, looseness=1.25] (645.center) to (646.center);
		\draw (650.center) to (651.center);
		\draw (652.center) to (651.center);
		\draw (650.center) to (653.center);
		\draw (653.center) to (652.center);
		\draw (665.center) to (655.center);
		\draw (668.center) to (669.center);
		\draw (688.center) to (689.center);
		\draw (690.center) to (691.center);
		\draw (692.center) to (691.center);
		\draw (690.center) to (693.center);
		\draw (693.center) to (692.center);
		\draw (697.center) to (698.center);
		\draw (701.center) to (702.center);
		\draw (703.center) to (704.center);
		\draw (736.center) to (737.center);
		\draw (738.center) to (737.center);
		\draw (736.center) to (739.center);
		\draw (739.center) to (738.center);
		\draw (743.center) to (744.center);
		\draw (748.center) to (747.center);
		\draw [style=red arrow] (751.center) to (752.center);
		\draw [bend right=90, looseness=1.50] (758.center) to (752.center);
		\draw (758.center) to (756.center);
		\draw [bend left=90, looseness=1.75] (756.center) to (751.center);
		\draw [bend right=270, looseness=1.25] (759.center) to (760.center);
		\draw (760.center) to (646.center);
		\draw [style=light red arrow] (645.center) to (759.center);
	\end{pgfonlayer}
\end{tikzpicture}
      \caption{Example of directional trace with $\vec u=(3,1)$. }
      \label{fig:traceu}
    \end{figure}

    We will usually apply it to rectangular shapes. When tracing a rectangle shape without specifying a direction, we will always imply that we sum opposite sides together:
    \[ \tr(R_{m,n}\cdot T) = \tr_{(0,n)}(\tr_{(m,0)}(R_{m,n}\cdot T)) \]

    \begin{center}
      \begin{tikzpicture}
	\begin{pgfonlayer}{nodelayer}
		\node [style=none] (515) at (19.5, -0.25) {};
		\node [style=none] (516) at (20.5, -0.25) {};
		\node [style=none] (517) at (20.5, -1.25) {};
		\node [style=none] (518) at (19.5, -1.25) {};
		\node [style=none] (520) at (19.5, -0.75) {};
		\node [style=none] (521) at (20, -0.25) {};
		\node [style=none] (522) at (20, -1.25) {};
		\node [style=none] (533) at (19.5, 1.25) {};
		\node [style=none] (534) at (20.5, 1.25) {};
		\node [style=none] (535) at (20.5, 0.25) {};
		\node [style=none] (536) at (19.5, 0.25) {};
		\node [style=none] (538) at (20, 1.25) {};
		\node [style=none] (540) at (20, 0.25) {};
		\node [style=none] (565) at (19.5, 0.75) {};
		\node [style=none] (571) at (20.5, 0.75) {};
		\node [style=none] (572) at (21, 0.75) {};
		\node [style=none] (573) at (20.5, -0.75) {};
		\node [style=none] (574) at (21, -0.75) {};
		\node [style=none] (575) at (13.25, 0) {$\mathrm{tr}(R_{3,2}\cdot T)$};
		\node [style=none] (576) at (15.5, 0) {$=$};
		\node [style=none] (615) at (21, -0.25) {};
		\node [style=none] (616) at (22, -0.25) {};
		\node [style=none] (617) at (22, -1.25) {};
		\node [style=none] (618) at (21, -1.25) {};
		\node [style=none] (621) at (21.5, -0.25) {};
		\node [style=none] (622) at (21.5, -1.25) {};
		\node [style=none] (624) at (21, 1.25) {};
		\node [style=none] (625) at (22, 1.25) {};
		\node [style=none] (626) at (22, 0.25) {};
		\node [style=none] (627) at (21, 0.25) {};
		\node [style=none] (629) at (21.5, 1.25) {};
		\node [style=none] (631) at (21.5, 0.25) {};
		\node [style=none] (634) at (22, 0.75) {};
		\node [style=none] (635) at (22.5, 0.75) {};
		\node [style=none] (636) at (22, -0.75) {};
		\node [style=none] (637) at (22.5, -0.75) {};
		\node [style=none] (646) at (18, -1.25) {};
		\node [style=none] (650) at (22.5, -0.25) {};
		\node [style=none] (651) at (23.5, -0.25) {};
		\node [style=none] (652) at (23.5, -1.25) {};
		\node [style=none] (653) at (22.5, -1.25) {};
		\node [style=none] (655) at (23, -0.25) {};
		\node [style=none] (658) at (22.5, 1.25) {};
		\node [style=none] (659) at (23.5, 1.25) {};
		\node [style=none] (660) at (23.5, 0.25) {};
		\node [style=none] (661) at (22.5, 0.25) {};
		\node [style=none] (663) at (23, 1.25) {};
		\node [style=none] (665) at (23, 0.25) {};
		\node [style=none] (666) at (23.5, 0.75) {};
		\node [style=none] (668) at (23.5, -0.75) {};
		\node [style=none] (748) at (23, -1.25) {};
		\node [style=none] (756) at (18.5, -1.25) {};
		\node [style=none] (758) at (18.5, 1.25) {};
		\node [style=none] (760) at (18, 1.25) {};
		\node [style=none] (761) at (19, -1.25) {};
		\node [style=none] (762) at (19, 1.25) {};
		\node [style=none] (763) at (19.5, 3.25) {};
		\node [style=none] (764) at (19.5, 2.75) {};
		\node [style=none] (765) at (23.5, 2.75) {};
		\node [style=none] (766) at (23.5, 3.25) {};
	\end{pgfonlayer}
	\begin{pgfonlayer}{edgelayer}
		\draw (515.center) to (516.center);
		\draw (517.center) to (516.center);
		\draw (515.center) to (518.center);
		\draw (518.center) to (517.center);
		\draw (533.center) to (534.center);
		\draw (535.center) to (534.center);
		\draw (533.center) to (536.center);
		\draw (536.center) to (535.center);
		\draw (540.center) to (521.center);
		\draw (571.center) to (572.center);
		\draw (573.center) to (574.center);
		\draw (615.center) to (616.center);
		\draw (617.center) to (616.center);
		\draw (615.center) to (618.center);
		\draw (618.center) to (617.center);
		\draw (624.center) to (625.center);
		\draw (626.center) to (625.center);
		\draw (624.center) to (627.center);
		\draw (627.center) to (626.center);
		\draw (631.center) to (621.center);
		\draw (634.center) to (635.center);
		\draw (636.center) to (637.center);
		\draw (650.center) to (651.center);
		\draw (652.center) to (651.center);
		\draw (650.center) to (653.center);
		\draw (653.center) to (652.center);
		\draw (658.center) to (659.center);
		\draw (660.center) to (659.center);
		\draw (658.center) to (661.center);
		\draw (661.center) to (660.center);
		\draw (665.center) to (655.center);
		\draw (758.center) to (756.center);
		\draw (760.center) to (646.center);
		\draw (762.center) to (761.center);
		\draw (765.center) to (764.center);
		\draw (766.center) to (763.center);
		\draw [bend left=90, looseness=2.75] (565.center) to (764.center);
		\draw [bend left=90, looseness=2.00] (520.center) to (763.center);
		\draw [bend left=105, looseness=1.25] (762.center) to (538.center);
		\draw [bend left=90, looseness=0.75] (758.center) to (629.center);
		\draw [bend left=90, looseness=0.75] (760.center) to (663.center);
		\draw [bend right=90, looseness=1.50] (761.center) to (522.center);
		\draw [bend right=90] (756.center) to (622.center);
		\draw [bend right=90] (646.center) to (748.center);
		\draw [bend right=90, looseness=2.00] (666.center) to (765.center);
		\draw [bend right=90, looseness=1.50] (668.center) to (766.center);
	\end{pgfonlayer}
\end{tikzpicture}
    \end{center}

	Notice that we have $\tr(R_{m,n}\cdot (T\uplus L))=\tr(R_{m,n}\cdot T)+ \tr(R_{m,n}\cdot L)$ and $\tr(R_{m,n}\cdot (T\times L))=\tr(R_{m,n}\cdot T) \tr(R_{m,n}\cdot L)$.
    Using this directional trace, we can now define the two notions of trace-aperiodicity. Note that intuitively, a \textit{strongly} aperiodic tileset is a tileset that cannot tile in a \textit{weakly} periodic manner. Similarly, a weakly aperiodic tileset cannot tile in a strongly periodic way (the two strongness notions switch when we consider the negative assertion).

    \begin{definition}[Strong trace aperiodicity]
      A tensorial $T$ is \emph{strongly trace-aperiodic} if for all $\vec u=(x,y) \neq 0$, there exists an $n$ such that either $\tr_{\vec u}(R_{x,n}\cdot T) = 0$ or $\tr_{\vec u}(R_{n,y}\cdot T) = 0$.
    \end{definition}

    \begin{definition}[Weak trace aperiodicity]
      A tensorial $T$ is \emph{weakly trace-aperiodic} if for all rectangle $R_{n,m}$, $\tr(R_{n,m}\cdot T) = 0$.
    \end{definition}

    Those definitions are generalizations of the usual ones for Wang tile sets. 

    First, we show that the link between the aperiodicity of a tensorial tile and the aperiodicity of its supports works as in dimension one.
    \begin{proposition}
            \label{prop:trivial}
      Given a tensorial tile $T$, if $\mathrm{supp}(T)$ is weakly (respectively strongly) aperiodic then $T$ is weakly (respectively strongly) trace aperiodic. Furthermore, the converse holds if $T$ is a possibilistic or probabilistic tile.
    \end{proposition}

    \begin{proof}
      We start by proving the contrapositive of the weak case. Let's consider a tensorial tile $T$ which is not strongly aperiodic, in other words, there is a $\vec u=(x,y) \neq 0$ such that for all $n$, $\tr_{\vec u}(R_{x,n}\cdot T) \neq 0$ and $\tr_{\vec u}(R_{n,y}\cdot T) \neq 0$, by symmetry we will consider that $x\neq 0$ and only consider the rectangles $R_{x,n}$. This implies that at least one of the coefficient in the sum defining the trace is non-negative, and then $\mathrm{supp}(T)$ can tile arbitrary wide stripes with the $u$-directional boundary conditions.
      
      \begin{center}
        \input{weakperiod.tikz}
      \end{center}
      
      By repeating those stripes vertically we get $u$-periodic tilings of arbitrarily wide shapes,
      and then a valid  $u$-periodic tiling of $\Z^2$ by compactness. So $\mathrm{supp}(T)$ is not strongly aperiodic.
      
      For the converse in the possibilistic and probabilistic case, if $\mathrm{supp}(T)$ admits a $u$-periodic tiling providing a non-null term in the sum defining the $u$-directional trace. If the coefficients of $T$ are non-negative, this implies that the $u$-directional trace is non-null as well, and then $T$ is not strongly trace-aperiodic.
      
      \smallskip
      
      We now prove the contrapositive of the weak case. Let's consider a tensorial tile $T$ which is not weakly aperiodic. In other words, there is a rectangle $R_{n,m}$ such that $\tr(R_{n,m}\cdot T) \neq 0$. It implies that at least one of the coefficients in the double sum defining the trace is non-zero. Then $\mathrm{supp}(T)$ can tile the rectangle with both vertical and horizontal periodic boundary conditions.
      
      \begin{center}
        \input{strongperiod.tikz}
      \end{center}
      
      Repeating this rectangle horizontally and vertically provides a valid tiling of $\Z^2$, which is both $(n,0)$-periodic and $(0,m)$-periodic, and then $\mathrm{supp}(T)$ is not weakly aperiodic.
      
      For the converse in the possibilistic and probabilistic case, if $\mathrm{supp}(T)$ has strongly periodic tilings, then it is known that $\mathrm{supp}(T)$ admits rectangular periods, meaning that there are two vectors $(n,0)$ and $(O,m)$ and a configuration which is both $(n,0)$-periodic and $(O,m)$-periodic. This tiling provides a way to tile the rectangle $R_{n,m}$ with both vertical and horizontal boundary conditions, and then there is a non-null term in the sum defining $\tr(R_{n,m}\cdot T)$.
      If the coefficients of $T$ are non-negative (possibilistic and probabilistic case) this implies that $\tr(R_{n,m}\cdot T) \neq 0 $, and then $T$ is not weakly trace-aperiodic.
    \end{proof}

    Using the previous results on dominoes one can show, as expected, that strong \emph{periodicity} (the negation of the previous definition) implies that the tile tiles the plane:

    \begin{lemma}\label{lemma:periodictiles}
    If a tensorial tile is strongly trace-periodic, then it tiles the plane.
    \end{lemma}
    
    \begin{proof}
    	Let $T$ be a strongly periodic trace-periodic tensorial tile. In other words, there exists $n$ and $m$ such that $\tr(R_{n,m}\cdot T) \neq 0 $. Seing $\tr_{(0,n)}(R_{n,m}\cdot T)$ as a matrix we have $\tr(R_{n,m}\cdot T)=\tr(\tr_{(0,n)}(R_{n,m}\cdot T))$. Then $\tr_{(0,n)}(R_{n,m}\cdot T)$ is a trace-periodic tensorial domino, and by Lemma \ref{lem:perpav}, one can find arbitrarily large k such that $\tr_{(0,n)}(R_{n,m}\cdot T)^k = \tr_{(0,n)}(R_{n,mk}\cdot T)\neq 0$. Reiterating the same argument in the other direction we have that we can find arbitrarily large $k$ and $l$ such that $\tr(R_{nl,mk}\cdot T) \neq 0$. It implies that $R_{nl,mk}\cdot T \neq 0$, so $T$ tiles the plane.
    \end{proof}

    As the name suggests, any strongly trace-aperiodic tensorial tile is also weakly trace-aperiodic. 
    Perhaps less straightforwardly, weak and strong aperiodicity are actually equivalent, just like in the classical case.

    \begin{proposition}
            \label{prop:2periodic}
      If a tensorial tile $T$ is weakly trace-aperiodic, then it is strongly trace-aperiodic.
    \end{proposition}
	
	The proof is more difficult than the well-known similar result in the classical case due to the possible presence of negative coefficients in the tensors. However, it follows similar ideas, but re-expressed in the new tensorial formalism.

    \begin{proof}
      We prove the contrapositive: assume that $T$ is not strongly aperiodic, i.e. there is a $\vec u=(a,b) \neq 0$ such that for all $n$, $\tr_{\vec u}(R_{a,n}\cdot T) \neq 0$ and $\tr_{\vec u}(R_{n,b}\cdot T) \neq 0$, without loss of generality we will assume that $b\neq 0$ and $a,b \geq 0$, the other cases being symmetrical. For any $n$, $\tr_{\vec u}(R_{n,b}\cdot T)$ is depicted:
      \begin{center}
        \begin{tikzpicture}
	\begin{pgfonlayer}{nodelayer}
		\node [style=none] (290) at (0, 0) {$\tr_{\vec u}(R_{n,b}\cdot T)=$};
		\node [style=none] (291) at (7, 2) {};
		\node [style=none] (292) at (11, 2) {};
		\node [style=none] (293) at (11, -2) {};
		\node [style=none] (294) at (7, -2) {};
		\node [style=none] (295) at (7.5, 2) {};
		\node [style=none] (296) at (7.5, 2.5) {};
		\node [style=none] (297) at (8.5, 2.5) {};
		\node [style=none] (298) at (8.5, 2) {};
		\node [style=none] (299) at (9.5, 2) {};
		\node [style=none] (301) at (10.5, 2) {};
		\node [style=none] (305) at (7.5, -2) {};
		\node [style=none] (306) at (8.5, -2) {};
		\node [style=none] (307) at (9.5, -2) {};
		\node [style=none] (308) at (10.5, -2) {};
		\node [style=none] (309) at (10.5, -2.5) {};
		\node [style=none] (310) at (9.5, -2.5) {};
		\node [style=none] (311) at (7, 0.5) {};
		\node [style=none] (312) at (7, -0.5) {};
		\node [style=none] (313) at (6.5, 0.5) {};
		\node [style=none] (314) at (6.5, -0.5) {};
		\node [style=none] (315) at (11, 0.5) {};
		\node [style=none] (316) at (11, -0.5) {};
		\node [style=none] (317) at (11.5, 0.5) {};
		\node [style=none] (318) at (11.5, -0.5) {};
		\node [style=none] (320) at (9, 0) {$R_{n,b}\cdot T$};
		\node [style=none] (321) at (5, 2) {};
		\node [style=none] (322) at (4, 2) {};
		\node [style=none] (323) at (4, -2) {};
		\node [style=none] (324) at (5, -2) {};
		\node [style=none] (325) at (6.25, -3.25) {};
		\node [style=none] (326) at (6.25, -4.25) {};
		\node [style=none] (327) at (7.25, 4) {};
		\node [style=none] (328) at (7.25, 5) {};
		\node [style=none] (329) at (8, 3) {};
		\node [style=none] (330) at (8.5, 3) {};
		\node [style=none] (331) at (7.5, 3) {};
		\node [style=none] (333) at (8, 3.5) {$a$};
		\node [style=none] (334) at (10, -3) {};
		\node [style=none] (335) at (10.5, -3) {};
		\node [style=none] (336) at (9.5, -3) {};
		\node [style=none] (339) at (10, 1.5) {};
		\node [style=none] (340) at (10.5, 1.5) {};
		\node [style=none] (341) at (9.5, 1.5) {};
		\node [style=none] (343) at (10, 1) {$n-a$};
		\node [style=none] (344) at (10, -3.5) {$a$};
		\node [style=none] (345) at (8, -1.5) {};
		\node [style=none] (346) at (8.5, -1.5) {};
		\node [style=none] (347) at (7.5, -1.5) {};
		\node [style=none] (348) at (8, -1) {$n-a$};
		\node [style=none] (349) at (12, 0) {};
		\node [style=none] (350) at (12, -0.5) {};
		\node [style=none] (351) at (12, 0.5) {};
		\node [style=none] (352) at (12.5, 0) {$b$};
		\node [style=none] (353) at (6, 0) {};
		\node [style=none] (354) at (6, -0.5) {};
		\node [style=none] (355) at (6, 0.5) {};
		\node [style=none] (356) at (5.5, 0) {$b$};
		\node [style=none] (357) at (4.5, 2) {$...$};
		\node [style=none] (358) at (8, 2.25) {$...$};
		\node [style=none] (359) at (10, 2.25) {$...$};
		\node [style=none] (360) at (10, -2.25) {$...$};
		\node [style=none] (361) at (8, -2.25) {$...$};
		\node [style=none] (362) at (4.5, -2) {$...$};
		\node [style=none] (363) at (7.25, 4.5) {$\bvdots$};
		\node [style=none] (367) at (6.25, -3.75) {$\bvdots$};
		\node [style=none] (368) at (6.75, 0) {$\bvdots$};
		\node [style=none] (369) at (11.25, 0) {$\bvdots$};
	\end{pgfonlayer}
	\begin{pgfonlayer}{edgelayer}
		\draw (291.center) to (292.center);
		\draw (292.center) to (293.center);
		\draw (293.center) to (294.center);
		\draw (294.center) to (291.center);
		\draw (296.center) to (295.center);
		\draw (297.center) to (298.center);
		\draw (315.center) to (317.center);
		\draw (318.center) to (316.center);
		\draw (308.center) to (309.center);
		\draw (307.center) to (310.center);
		\draw (314.center) to (312.center);
		\draw (311.center) to (313.center);
		\draw [bend left=315] (328.center) to (322.center);
		\draw (322.center) to (323.center);
		\draw [bend right=45] (323.center) to (326.center);
		\draw [bend right=45] (326.center) to (306.center);
		\draw [bend left=45] (305.center) to (325.center);
		\draw [bend left=45] (325.center) to (324.center);
		\draw (324.center) to (321.center);
		\draw [bend left=45] (321.center) to (327.center);
		\draw [bend left=45] (327.center) to (299.center);
		\draw [bend right=45] (301.center) to (328.center);
		\draw [style=arrow] (329.center) to (330.center);
		\draw [style=arrow] (329.center) to (331.center);
		\draw [style=arrow] (334.center) to (335.center);
		\draw [style=arrow] (334.center) to (336.center);
		\draw [style=arrow] (339.center) to (340.center);
		\draw [style=arrow] (339.center) to (341.center);
		\draw [style=arrow] (345.center) to (346.center);
		\draw [style=arrow] (345.center) to (347.center);
		\draw [style=arrow] (349.center) to (350.center);
		\draw [style=arrow] (349.center) to (351.center);
		\draw [style=arrow] (353.center) to (354.center);
		\draw [style=arrow] (353.center) to (355.center);
	\end{pgfonlayer}
\end{tikzpicture}
      \end{center}

      For all $n\geq a+1 $ we define the matrix $K_n $ by reorganizing the indices of $\tr_{\vec u}(R_{n,b}\cdot T)$ as follows:

      \begin{center}
        \begin{tikzpicture}
	\begin{pgfonlayer}{nodelayer}
		\node [style=none] (290) at (0, 0) {$K_n =$};
		\node [style=none] (291) at (6, 0.5) {};
		\node [style=none] (292) at (10, 0.5) {};
		\node [style=none] (293) at (10, -3.5) {};
		\node [style=none] (294) at (6, -3.5) {};
		\node [style=none] (295) at (6.5, 0.5) {};
		\node [style=none] (298) at (7.5, 0.5) {};
		\node [style=none] (299) at (8.5, 0.5) {};
		\node [style=none] (301) at (9.5, 0.5) {};
		\node [style=none] (305) at (6.5, -3.5) {};
		\node [style=none] (306) at (7.5, -3.5) {};
		\node [style=none] (311) at (6, -1) {};
		\node [style=none] (312) at (6, -2) {};
		\node [style=none] (313) at (2, -1) {};
		\node [style=none] (314) at (2, -2) {};
		\node [style=none] (315) at (10, -1) {};
		\node [style=none] (316) at (10, -2) {};
		\node [style=none] (317) at (13, -1) {};
		\node [style=none] (318) at (13, -2) {};
		\node [style=none] (320) at (8, -1.5) {$R_{n,b}\cdot T$};
		\node [style=none] (321) at (4, 0.5) {};
		\node [style=none] (322) at (3, 0.5) {};
		\node [style=none] (323) at (3, -3.5) {};
		\node [style=none] (324) at (4, -3.5) {};
		\node [style=none] (325) at (5.25, -4.75) {};
		\node [style=none] (326) at (5.25, -5.75) {};
		\node [style=none] (327) at (6.25, 2.5) {};
		\node [style=none] (328) at (6.25, 3.5) {};
		\node [style=none] (349) at (13.5, -1.5) {};
		\node [style=none] (350) at (13.5, -2) {};
		\node [style=none] (351) at (13.5, -1) {};
		\node [style=none] (352) at (14, -1.5) {$b$};
		\node [style=none] (353) at (1.5, -1.5) {};
		\node [style=none] (354) at (1.5, -2) {};
		\node [style=none] (355) at (1.5, -1) {};
		\node [style=none] (356) at (1, -1.5) {$b$};
		\node [style=none] (357) at (3.5, 0.5) {$...$};
		\node [style=none] (358) at (7, 0.75) {$...$};
		\node [style=none] (359) at (9, 0.75) {$...$};
		\node [style=none] (360) at (9, -3.75) {$...$};
		\node [style=none] (361) at (7, -3.75) {$...$};
		\node [style=none] (362) at (3.5, -3.5) {$...$};
		\node [style=none] (363) at (6.25, 3) {$\bvdots$};
		\node [style=none] (367) at (5.25, -5.25) {$\bvdots$};
		\node [style=none] (368) at (5.75, -1.5) {$\bvdots$};
		\node [style=none] (369) at (10.25, -1.5) {$\bvdots$};
		\node [style=none] (370) at (2, 2) {};
		\node [style=none] (371) at (2, 1) {};
		\node [style=none] (372) at (6, 2) {};
		\node [style=none] (373) at (6, 1) {};
		\node [style=none] (374) at (6, 1.5) {$\bvdots$};
		\node [style=none] (375) at (2.25, 1.5) {$\bvdots$};
		\node [style=none] (376) at (2.25, -1.5) {$\bvdots$};
		\node [style=none] (377) at (1.5, 1.5) {};
		\node [style=none] (378) at (1.5, 1) {};
		\node [style=none] (379) at (1.5, 2) {};
		\node [style=none] (380) at (1, 1.5) {$a$};
		\node [style=none] (383) at (13, 2) {};
		\node [style=none] (384) at (13, 1) {};
		\node [style=none] (385) at (12.75, 1.5) {$\bvdots$};
		\node [style=none] (387) at (13.5, 1.5) {};
		\node [style=none] (388) at (13.5, 1) {};
		\node [style=none] (389) at (13.5, 2) {};
		\node [style=none] (390) at (14, 1.5) {$a$};
		\node [style=none] (392) at (9.5, -3.5) {};
		\node [style=none] (393) at (8.5, -3.5) {};
		\node [style=none] (395) at (10, -5) {};
		\node [style=none] (396) at (10, -4) {};
		\node [style=none] (397) at (10, -4.5) {$\bvdots$};
		\node [style=none] (404) at (10.5, -3.5) {};
		\node [style=none] (405) at (11.5, -3.5) {};
		\node [style=none] (406) at (10, -5) {};
		\node [style=none] (407) at (10, -4) {};
		\node [style=none] (408) at (11.5, 0.5) {};
		\node [style=none] (409) at (10.5, 0.5) {};
		\node [style=none] (410) at (12, 2) {};
		\node [style=none] (411) at (12, 1) {};
		\node [style=none] (412) at (12.75, -1.5) {$\bvdots$};
		\node [style=none] (413) at (9, 0) {};
		\node [style=none] (414) at (9.5, 0) {};
		\node [style=none] (415) at (8.5, 0) {};
		\node [style=none] (416) at (9, -0.5) {$n-a$};
		\node [style=none] (417) at (7, -3) {};
		\node [style=none] (418) at (7.5, -3) {};
		\node [style=none] (419) at (6.5, -3) {};
		\node [style=none] (420) at (7, -2.5) {$n-a$};
	\end{pgfonlayer}
	\begin{pgfonlayer}{edgelayer}
		\draw (291.center) to (292.center);
		\draw (292.center) to (293.center);
		\draw (293.center) to (294.center);
		\draw (294.center) to (291.center);
		\draw (315.center) to (317.center);
		\draw (318.center) to (316.center);
		\draw (314.center) to (312.center);
		\draw (311.center) to (313.center);
		\draw [bend left=315] (328.center) to (322.center);
		\draw (322.center) to (323.center);
		\draw [bend right=45] (323.center) to (326.center);
		\draw [bend right=45] (326.center) to (306.center);
		\draw [bend left=45] (305.center) to (325.center);
		\draw [bend left=45] (325.center) to (324.center);
		\draw (324.center) to (321.center);
		\draw [bend left=45] (321.center) to (327.center);
		\draw [bend left=45] (327.center) to (299.center);
		\draw [bend right=45] (301.center) to (328.center);
		\draw [style=arrow] (349.center) to (350.center);
		\draw [style=arrow] (349.center) to (351.center);
		\draw [style=arrow] (353.center) to (354.center);
		\draw [style=arrow] (353.center) to (355.center);
		\draw [bend left=45] (373.center) to (295.center);
		\draw [bend left=45] (372.center) to (298.center);
		\draw (372.center) to (370.center);
		\draw (371.center) to (373.center);
		\draw [style=arrow] (377.center) to (378.center);
		\draw [style=arrow] (377.center) to (379.center);
		\draw [style=arrow] (387.center) to (388.center);
		\draw [style=arrow] (387.center) to (389.center);
		\draw [bend left=45] (396.center) to (392.center);
		\draw [bend left=45] (395.center) to (393.center);
		\draw [bend right=45] (407.center) to (404.center);
		\draw [bend right=45] (406.center) to (405.center);
		\draw [bend right=45] (411.center) to (408.center);
		\draw [bend right=45] (410.center) to (409.center);
		\draw (409.center) to (404.center);
		\draw (405.center) to (408.center);
		\draw (410.center) to (383.center);
		\draw (384.center) to (411.center);
		\draw [style=arrow] (413.center) to (414.center);
		\draw [style=arrow] (413.center) to (415.center);
		\draw [style=arrow] (417.center) to (418.center);
		\draw [style=arrow] (417.center) to (419.center);
	\end{pgfonlayer}
\end{tikzpicture}
      \end{center}

      We remark that $K_m \circ K_n = K_{n+m}  $:
      \begin{center}
        \scalebox{0.9}{\input{compK.tikz}}
      \end{center}

      As $K_n $ is $\tr_{\vec u}(R_{n,b}\cdot T)$ with reorganized indices, $K_n =0 $ only if  
      $\tr_{\vec u}(R_{n,b}\cdot T)=0 $ which is never true by hypothesis, and then for each $n$, $K_n \neq 0 $. In particular $(K_n)^l = K_{nl} \neq 0 $ for all $l$, so $ K_{n} $ is not nilpotent.
      Setting $n=a+1$, by charaterization of (non-)nilpotent matrices (\cref{trace}), there exists $m_0$ such that $\tr_{(m_0(a+1), 0)}\left( \tr_{(a,b)}(R_{m_0(a+1),b}) \right) = \tr(K^{m_0(a+1)}) \neq 0$. We will fix $q=m_0(a+1)$, graphically:
      \begin{center}
        \begin{tikzpicture}
	\begin{pgfonlayer}{nodelayer}
		\node [style=none] (290) at (0, 0) {$\tr( K_q ) =$};
		\node [style=none] (291) at (11.5, -1) {};
		\node [style=none] (292) at (15.5, -1) {};
		\node [style=none] (293) at (15.5, -5) {};
		\node [style=none] (294) at (11.5, -5) {};
		\node [style=none] (295) at (12, -1) {};
		\node [style=none] (298) at (13, -1) {};
		\node [style=none] (299) at (14, -1) {};
		\node [style=none] (301) at (15, -1) {};
		\node [style=none] (305) at (12, -5) {};
		\node [style=none] (306) at (13, -5) {};
		\node [style=none] (311) at (11.5, -2.5) {};
		\node [style=none] (312) at (11.5, -3.5) {};
		\node [style=none] (313) at (7.5, -2.5) {};
		\node [style=none] (314) at (7.5, -3.5) {};
		\node [style=none] (315) at (15.5, -2.5) {};
		\node [style=none] (316) at (15.5, -3.5) {};
		\node [style=none] (317) at (17.5, -2.5) {};
		\node [style=none] (318) at (17.5, -3.5) {};
		\node [style=none] (320) at (13.5, -3) {$R_{q,b}\cdot T$};
		\node [style=none] (321) at (9.5, -1) {};
		\node [style=none] (322) at (8.5, -1) {};
		\node [style=none] (323) at (8.5, -5) {};
		\node [style=none] (324) at (9.5, -5) {};
		\node [style=none] (325) at (10.75, -6.25) {};
		\node [style=none] (326) at (10.75, -7.25) {};
		\node [style=none] (327) at (11.75, 1) {};
		\node [style=none] (328) at (11.75, 2) {};
		\node [style=none] (357) at (9, -1) {$...$};
		\node [style=none] (358) at (12.5, -0.75) {$...$};
		\node [style=none] (359) at (14.5, -0.75) {$...$};
		\node [style=none] (360) at (14.5, -5.25) {$...$};
		\node [style=none] (361) at (12.5, -5.25) {$...$};
		\node [style=none] (362) at (9, -5) {$...$};
		\node [style=none] (363) at (11.75, 1.5) {$\bvdots$};
		\node [style=none] (367) at (10.75, -6.75) {$\bvdots$};
		\node [style=none] (368) at (11.25, -3) {$\bvdots$};
		\node [style=none] (369) at (15.75, -3) {$\bvdots$};
		\node [style=none] (370) at (7.5, 0.5) {};
		\node [style=none] (371) at (7.5, -0.5) {};
		\node [style=none] (372) at (11.5, 0.5) {};
		\node [style=none] (373) at (11.5, -0.5) {};
		\node [style=none] (374) at (11.5, 0) {$\bvdots$};
		\node [style=none] (375) at (7.5, 0) {$\bvdots$};
		\node [style=none] (392) at (15, -5) {};
		\node [style=none] (393) at (14, -5) {};
		\node [style=none] (395) at (15.5, -6.5) {};
		\node [style=none] (396) at (15.5, -5.5) {};
		\node [style=none] (397) at (15.5, -6) {$\bvdots$};
		\node [style=none] (404) at (16, -5) {};
		\node [style=none] (405) at (17, -5) {};
		\node [style=none] (406) at (15.5, -6.5) {};
		\node [style=none] (407) at (15.5, -5.5) {};
		\node [style=none] (408) at (17, -1) {};
		\node [style=none] (409) at (16, -1) {};
		\node [style=none] (410) at (17.5, 0.5) {};
		\node [style=none] (411) at (17.5, -0.5) {};
		\node [style=none] (413) at (14.5, -1.5) {};
		\node [style=none] (414) at (15, -1.5) {};
		\node [style=none] (415) at (14, -1.5) {};
		\node [style=none] (416) at (14.5, -2) {$q-a$};
		\node [style=none] (417) at (12.5, -4.5) {};
		\node [style=none] (418) at (13, -4.5) {};
		\node [style=none] (419) at (12, -4.5) {};
		\node [style=none] (420) at (12.5, -4) {$q-a$};
		\node [style=none] (435) at (17.5, 3) {$\bvdots$};
		\node [style=none] (436) at (17.5, 0) {$\bvdots$};
		\node [style=none] (437) at (17.5, 0.5) {};
		\node [style=none] (438) at (17.5, -0.5) {};
		\node [style=none] (439) at (17.5, 3.5) {};
		\node [style=none] (440) at (17.5, 2.5) {};
		\node [style=none] (441) at (18.5, 1.5) {};
		\node [style=none] (442) at (19.5, 1.5) {};
		\node [style=none] (443) at (19, 1.5) {$...$};
		\node [style=none] (448) at (7.5, 3) {$\bvdots$};
		\node [style=none] (450) at (7.5, 0.5) {};
		\node [style=none] (451) at (7.5, -0.5) {};
		\node [style=none] (452) at (7.5, 3.5) {};
		\node [style=none] (453) at (7.5, 2.5) {};
		\node [style=none] (454) at (6.5, 1.5) {};
		\node [style=none] (455) at (5.5, 1.5) {};
		\node [style=none] (456) at (6, 1.5) {$...$};
		\node [style=none] (457) at (7.5, -2.5) {};
		\node [style=none] (458) at (7.5, -3.5) {};
		\node [style=none] (459) at (7.5, 5.5) {};
		\node [style=none] (460) at (7.5, 4.5) {};
		\node [style=none] (461) at (7.5, 5) {$\bvdots$};
		\node [style=none] (462) at (7.5, -3) {$\bvdots$};
		\node [style=none] (463) at (7.5, -2.5) {};
		\node [style=none] (464) at (7.5, -3.5) {};
		\node [style=none] (465) at (7.5, 5.5) {};
		\node [style=none] (466) at (7.5, 4.5) {};
		\node [style=none] (467) at (4, 1) {};
		\node [style=none] (468) at (3, 1) {};
		\node [style=none] (469) at (3.5, 1) {$...$};
		\node [style=none] (470) at (17.5, -2.5) {};
		\node [style=none] (471) at (17.5, -3.5) {};
		\node [style=none] (472) at (17.5, 5.5) {};
		\node [style=none] (473) at (17.5, 4.5) {};
		\node [style=none] (474) at (17.5, 5) {$\bvdots$};
		\node [style=none] (475) at (17.5, -3) {$\bvdots$};
		\node [style=none] (476) at (17.5, -2.5) {};
		\node [style=none] (477) at (17.5, -3.5) {};
		\node [style=none] (478) at (17.5, 5.5) {};
		\node [style=none] (479) at (17.5, 4.5) {};
		\node [style=none] (480) at (21, 1) {};
		\node [style=none] (481) at (22, 1) {};
		\node [style=none] (482) at (21.5, 1) {$...$};
		\node [style=none] (483) at (23, 0) {$\neq 0$};
	\end{pgfonlayer}
	\begin{pgfonlayer}{edgelayer}
		\draw (291.center) to (292.center);
		\draw (292.center) to (293.center);
		\draw (293.center) to (294.center);
		\draw (294.center) to (291.center);
		\draw (315.center) to (317.center);
		\draw (318.center) to (316.center);
		\draw (314.center) to (312.center);
		\draw (311.center) to (313.center);
		\draw [bend left=315] (328.center) to (322.center);
		\draw (322.center) to (323.center);
		\draw [bend right=45] (323.center) to (326.center);
		\draw [bend right=45] (326.center) to (306.center);
		\draw [bend left=45] (305.center) to (325.center);
		\draw [bend left=45] (325.center) to (324.center);
		\draw (324.center) to (321.center);
		\draw [bend left=45] (321.center) to (327.center);
		\draw [bend left=45] (327.center) to (299.center);
		\draw [bend right=45] (301.center) to (328.center);
		\draw [bend left=45] (373.center) to (295.center);
		\draw [bend left=45] (372.center) to (298.center);
		\draw (372.center) to (370.center);
		\draw (371.center) to (373.center);
		\draw [bend left=45] (396.center) to (392.center);
		\draw [bend left=45] (395.center) to (393.center);
		\draw [bend right=45] (407.center) to (404.center);
		\draw [bend right=45] (406.center) to (405.center);
		\draw [bend right=45] (411.center) to (408.center);
		\draw [bend right=45] (410.center) to (409.center);
		\draw (409.center) to (404.center);
		\draw (405.center) to (408.center);
		\draw [style=arrow] (413.center) to (414.center);
		\draw [style=arrow] (413.center) to (415.center);
		\draw [style=arrow] (417.center) to (418.center);
		\draw [style=arrow] (417.center) to (419.center);
		\draw [bend left=45] (439.center) to (442.center);
		\draw [bend left=45] (442.center) to (438.center);
		\draw [bend right=45] (437.center) to (441.center);
		\draw [bend right=45] (441.center) to (440.center);
		\draw [bend right=45] (452.center) to (455.center);
		\draw [bend right=45] (455.center) to (451.center);
		\draw [bend left=45] (450.center) to (454.center);
		\draw [bend left=45] (454.center) to (453.center);
		\draw [bend right=45] (465.center) to (468.center);
		\draw [bend right=45] (468.center) to (464.center);
		\draw [bend left=45] (463.center) to (467.center);
		\draw [bend left=45] (467.center) to (466.center);
		\draw [bend left=45] (478.center) to (481.center);
		\draw [bend left=45] (481.center) to (477.center);
		\draw [bend right=45] (476.center) to (480.center);
		\draw [bend right=45] (480.center) to (479.center);
		\draw (465.center) to (478.center);
		\draw (479.center) to (466.center);
		\draw (452.center) to (439.center);
		\draw (440.center) to (453.center);
	\end{pgfonlayer}
\end{tikzpicture}
      \end{center}

      We will now look at our tensors as matrices from bottom to top. Let $\sigma $ be the cyclic permutation matrix corresponding to moving the first $a$ elements to the end of a list of $q$ elements. We define the matrix $L$ as $L= \sigma \circ \tr_{(q, 0)}\left( R_{q,b} \right)$, pictorially:
      \begin{center}
        \begin{tikzpicture}
	\begin{pgfonlayer}{nodelayer}
		\node [style=none] (290) at (0, 0) {$L =$};
		\node [style=none] (291) at (4, 1) {};
		\node [style=none] (292) at (8, 1) {};
		\node [style=none] (293) at (8, -1) {};
		\node [style=none] (294) at (4, -1) {};
		\node [style=none] (295) at (4.5, 1) {};
		\node [style=none] (298) at (5.5, 1) {};
		\node [style=none] (299) at (6.5, 1) {};
		\node [style=none] (301) at (7.5, 1) {};
		\node [style=none] (305) at (4.5, -1) {};
		\node [style=none] (306) at (5.5, -1) {};
		\node [style=none] (317) at (8, -2) {};
		\node [style=none] (318) at (8, -3) {};
		\node [style=none] (320) at (6, 0) {$R_{q,b}\cdot T$};
		\node [style=none] (321) at (7.5, 4) {};
		\node [style=none] (322) at (6.5, 4) {};
		\node [style=none] (323) at (4.5, 1) {};
		\node [style=none] (324) at (5.5, 1) {};
		\node [style=none] (358) at (5, 1.25) {$...$};
		\node [style=none] (359) at (7, 1.25) {$...$};
		\node [style=none] (360) at (7, -1.25) {$...$};
		\node [style=none] (361) at (5, -1.25) {$...$};
		\node [style=none] (392) at (7.5, -1) {};
		\node [style=none] (393) at (6.5, -1) {};
		\node [style=none] (413) at (7, 4.5) {};
		\node [style=none] (414) at (7.5, 4.5) {};
		\node [style=none] (415) at (6.5, 4.5) {};
		\node [style=none] (416) at (7, 5) {$a$};
		\node [style=none] (470) at (8, -2) {};
		\node [style=none] (471) at (8, -3) {};
		\node [style=none] (472) at (8, 0.5) {};
		\node [style=none] (473) at (8, -0.5) {};
		\node [style=none] (474) at (8.25, 0) {$\bvdots$};
		\node [style=none] (475) at (8.25, -2.5) {$\bvdots$};
		\node [style=none] (476) at (8, -2) {};
		\node [style=none] (477) at (8, -3) {};
		\node [style=none] (478) at (8, 0.5) {};
		\node [style=none] (479) at (8, -0.5) {};
		\node [style=none] (480) at (8.75, -1.25) {};
		\node [style=none] (481) at (9.75, -1.25) {};
		\node [style=none] (482) at (9.25, -1.25) {$...$};
		\node [style=none] (484) at (5.5, 4) {};
		\node [style=none] (485) at (4.5, 4) {};
		\node [style=none] (486) at (5, 3.75) {$...$};
		\node [style=none] (487) at (7, 3.75) {$...$};
		\node [style=none] (488) at (5, 4.5) {};
		\node [style=none] (489) at (5.5, 4.5) {};
		\node [style=none] (490) at (4.5, 4.5) {};
		\node [style=none] (491) at (5, 5) {$q-a$};
		\node [style=none] (496) at (4.5, -4) {};
		\node [style=none] (497) at (5.5, -4) {};
		\node [style=none] (498) at (6.5, -4) {};
		\node [style=none] (499) at (7.5, -4) {};
		\node [style=none] (500) at (4, -2) {};
		\node [style=none] (501) at (4, -3) {};
		\node [style=none] (502) at (4, -2) {};
		\node [style=none] (503) at (4, -3) {};
		\node [style=none] (504) at (4, 0.5) {};
		\node [style=none] (505) at (4, -0.5) {};
		\node [style=none] (506) at (3.75, 0) {$\bvdots$};
		\node [style=none] (507) at (3.75, -2.5) {$\bvdots$};
		\node [style=none] (508) at (4, -2) {};
		\node [style=none] (509) at (4, -3) {};
		\node [style=none] (510) at (4, 0.5) {};
		\node [style=none] (511) at (4, -0.5) {};
		\node [style=none] (512) at (3.25, -1.25) {};
		\node [style=none] (513) at (2.25, -1.25) {};
		\node [style=none] (514) at (2.75, -1.25) {$...$};
		\node [style=none] (515) at (7, -3.75) {$...$};
		\node [style=none] (516) at (5, -3.75) {$...$};
		\node [style=none] (517) at (7, -4.5) {};
		\node [style=none] (518) at (7.5, -4.5) {};
		\node [style=none] (519) at (6.5, -4.5) {};
		\node [style=none] (520) at (7, -5) {$a$};
		\node [style=none] (521) at (5, -4.5) {};
		\node [style=none] (522) at (5.5, -4.5) {};
		\node [style=none] (523) at (4.5, -4.5) {};
		\node [style=none] (524) at (5, -5) {$q-a$};
		\node [style=none] (525) at (3.5, 3.5) {};
		\node [style=none] (526) at (3.5, 1.5) {};
		\node [style=grey text] (528) at (3, 2.5) {$\sigma$};
		\node [style=none] (529) at (8.5, 3.5) {};
		\node [style=none] (531) at (8.5, 1.5) {};
	\end{pgfonlayer}
	\begin{pgfonlayer}{edgelayer}
		\draw (291.center) to (292.center);
		\draw (292.center) to (293.center);
		\draw (293.center) to (294.center);
		\draw (294.center) to (291.center);
		\draw [in=90, out=-90] (322.center) to (323.center);
		\draw [in=-90, out=90] (324.center) to (321.center);
		\draw [style=arrow] (413.center) to (414.center);
		\draw [style=arrow] (413.center) to (415.center);
		\draw [bend left=45] (478.center) to (481.center);
		\draw [bend left=45] (481.center) to (477.center);
		\draw [bend right=45] (476.center) to (480.center);
		\draw [bend right=45] (480.center) to (479.center);
		\draw [in=-90, out=90] (301.center) to (484.center);
		\draw [in=90, out=-90] (485.center) to (299.center);
		\draw [style=arrow] (488.center) to (489.center);
		\draw [style=arrow] (488.center) to (490.center);
		\draw (305.center) to (496.center);
		\draw (497.center) to (306.center);
		\draw (393.center) to (498.center);
		\draw (392.center) to (499.center);
		\draw [bend right=45] (510.center) to (513.center);
		\draw [bend right=45] (513.center) to (509.center);
		\draw [bend left=45] (508.center) to (512.center);
		\draw [bend left=45] (512.center) to (511.center);
		\draw (476.center) to (508.center);
		\draw (509.center) to (477.center);
		\draw [style=arrow] (517.center) to (518.center);
		\draw [style=arrow] (517.center) to (519.center);
		\draw [style=arrow] (521.center) to (522.center);
		\draw [style=arrow] (521.center) to (523.center);
		\draw [style=pointille] (525.center) to (526.center);
		\draw [style=pointille, in=180, out=0] (525.center) to (529.center);
		\draw [style=pointille] (529.center) to (531.center);
		\draw [style=pointille] (531.center) to (526.center);
	\end{pgfonlayer}
\end{tikzpicture}
      \end{center}

      We directly see graphically that: $\tr(L)= \tr(K^{q}) \neq 0 $ so $L$ is not nilpotent. Then:

      \[\left(\sigma^l \circ \tr_{(q, 0)}\left( R_{q,lb}\cdot T \right)\right) \circ \left(\sigma^k \circ \tr_{(q, 0)}\left( R_{q,kb}\cdot T \right)\right)=\sigma^{(k+l)} \circ \tr_{(q, 0)}\left( R_{q,(k+l)b}\cdot T \right)\]

      We show it pictorially:

      \begin{center}
        \input{compL.tikz}
      \end{center}

      In this last step, we make the vertical columns of $T$ slide along the horizontal trace until the lower permutation $\sigma^k $ is completely disentangled and equal to the identity. This has the effect of composing the upper permutation with $\sigma^k $. Notice that the fact that $\sigma^k $ is cyclic is crucial for this operation to be done only by sliding columns along the trace. Then recombining the $T$s gives:
      \begin{center}
        \input{compL2.tikz}
      \end{center}
      From this it follows that: $L^k = \sigma^k \circ \tr_{(q, 0)}\left( R_{q,kb}\cdot T \right)$. Since $\sigma$ is a permutation there is a $p$ such that $\sigma^p =id $ and then $L^p = \tr_{(q, 0)}\left( R_{q,pb}\cdot T \right)$. Furthermore as $L$ is not nilpotent, $L^p $ is not nilpotent, and then by \cref{trace} there exists $m\in\N$ such that $\tr(L^{mp} )\neq 0 $. This provides a rectangle $R_{q,mpb}$ such that:
      \[ \tr(R_{q,mpb}\cdot T)=\tr_{(0,pb)}\left(\tr_{(q, 0)}\left( R_{q,mpb} \right)\right)=\tr(L^{mp} )\neq 0 \]
      and then $T$ is not weakly aperiodic.
    \end{proof}

    Overall, we see that just like in the possibilistic case, strong and weak aperiodicity are totally equivalent. Thus, we will only talk about trace aperiodicity without mentioning its weak or strong character. Notice that for classical tile sets having a rectangular strong period implies that there is a square strong period, however we did not manage to obtain a similar result for tensorial tiles in general. 
    Neither did we manage to find any counter example. 
    We state this problem as a first open problem.

    \begin{open}\label{open1}
      Given any tensorial tile $T$, does $\tr(R_{a,b} \cdot T)\neq 0$, for some $a,b\in \mathbb{N}$, implies that there is an $n\in \mathbb{N}$ such that $\tr(R_{n,n} \cdot T)\neq 0$ ?
    \end{open}
    
    This problem is more difficult than it seems. Given any rectangle such that $\tr(R_{n,m} \cdot T)\neq 0$ we can show that the sequences $k \mapsto \tr(R_{kn,m} \cdot T)$ and $k \mapsto \tr(R_{n,km} \cdot T)$ are linear recurrent, and then we can use a wide variety of results to independently study there zeroes. However, we need to understand the interplay of those different sequences to tell anything about $k \mapsto \tr(R_{k,k} \cdot T)$, it seems that the mathematical tools to describe those kind of two dimensional sequences are yet to be developed.

  \subsection{A quantum aperiodic tiletset}
    In dimension one, the only trace-aperiodic quantum tilesets are the ones that do not tile the line.
    In dimension two, just like in the classical case, these two notions are not equivalent.
    A classical aperiodic tileset provides an obvious (but not very satisfying) example of a quantum aperiodic tileset that tiles the plane. 
    The interesting question is whether quantum interference allow new \emph{kinds} of aperiodicites to appear. 
    Thus we have to look for a quantum tileset that is trace-aperiodic, but whose support is not aperiodic. Such a tileset really needs the interference to be aperiodic.

    \begin{theorem}\label{th:ap}
      There exists a quantum tile $T$ that: \textbf{has a non-aperiodic support}, \textbf{tiles the plane}, \textit{i.e.}, for all shape $S$, $S\cdot T \neq 0$, and \textbf{is trace-aperiodic}, \textit{i.e.}, for any rectangle $R, \tr(R\cdot T) = 0$.
    \end{theorem}

    \begin{proof}
      Let $\tau_A$ be any usual aperiodic tileset (for example Kari's aperiodic tileset \cite{KariAperiodic}) and $T_A$ its associated possibilistic tile.
      Let $N=\frac{1}{2}\begin{pmatrix}
                        1&1\\ -1& -1
                      \end{pmatrix}$ 
                      and $T_N$ be the following 2D quantum tile:
                      \[ (T_N)_{x,y,z,t} = \bra{x} N \ket{z} \bra{t} N \ket{y} .\]
                      We denote its support (wich is the full tileset) by $\tau_N$.
                      First, remark that since $N^2=0$, we have $R\cdot T_N = 0$ for any rectangle $R$, strictly larger than $1\times 1$. Then for any rectangle $R$,  $\tr(R\cdot T_{N}) = 0$, as $\tr(T_{N}) = \tr(N)^2 = 0$. However, all tiles are in $\tau_N$, in particular the ones with the same color everywhere, so $\tau_N$ has many valid periodic tilings.
                      
                      Then $T:= T_A \uplus T_N$ checks the conditions of the theorem:
                      \begin{description}
                        \item [Has a non-aperiodic support]: $\mathrm{supp}(T)=\mathrm{supp}(T_A \uplus T_N )=\tau_A \uplus \tau_N $ and $\tau_P$ has valid periodic configurations.
                        \item [Tiles the plane]: First $T\neq 0$ and for any rectangle $R$ bigger than $1\times 1$, $R\cdot T = R\cdot T_{A} \uplus R\cdot T_{N} = R\cdot T_{A} \oplus 0 \neq 0$ as $T_A$ tiles the plane. 
                        \item [Is trace-aperiodic]: For any rectangle $R$, $\tr(R\cdot T) = \tr(R\cdot T_{A}) + \tr(R\cdot T_{N}) = 0$ because $T_A$ is aperiodic and we already know that $\tr(R\cdot T_{N}) = 0$.
                      \end{description}
     \end{proof}

     This very simple example shows that quantum interference can force aperiodicity, however it is quite artificial, since a subset of its support form a classical aperiodic tileset (the support of $T_A$ is classicaly apeirodic by definition). This suggest the following definition.
     
     \begin{definition}[Purely quantum aperiodic tile]
     	A tensorial tile $T$ is said \textbf{purely quantum aperiodic} if it tiles the plane, it is trace-aperiodic and no subsets of $\mathrm{supp}(T)$ are aperiodic tilesets tiling the plane.
     \end{definition}
	
	Such tileset would exhibit aperiodicity of a new kind that do not rely on classical aperiodicity but only on interference. The existence of such a tileset is the second open problem of quantum tiling theory that we state.
	
	\begin{open}
		Does a purely quantum aperiodic tensorial tile exist ?
	\end{open}

	In particular, such tileset would not be concerned by the lower bound of \cite{Jeandel_Rao_2021}, opening the possibility of quantum aperiodic tilesets with less than 11 tiles.

\section{Applications}
  \label{sec:appli}

  \subsection{Space time diagram of quantum cellular automata}
    Wang tilesets are very often used as a way of representing space-time diagram of (classical) cellular automata. In this section we show that our model of quantum tiles allow us to do the same for certain type of quantum cellular automata.

    Without defining all the term, a (one dimensional) quantum cellular automaton is an operator on quantum configurations which is unitary, shift-invariant and causal (meaning that it can be decomposed into a uniform local functions with a finite radius of influence).
    A particular type of quantum cellular automaton is a \emph{partitioned (one-dimensional) quantum cellular automaton} (PQCA), defined by a local unitary operator $U$, which can be written as a 4-tensor. It is also required to preserve a special quiescent state $\ket q$, i.e. $U\ket q = \ket q$. 
    Its global operator is then $\bigotimes_{Z\Z^n} U$ (\cref{fig:PQCA}).
    This type of quantum cellular automaton may seem restrictive, as its radius is only two cells, however they are intrinsically universal: they able to simulate any quantum cellular automaton \cite{Arrighi_Grattage_2012}.

    \begin{wrapfigure}{r}{0.4\textwidth}
      \centering
      \includegraphics[width=5cm]{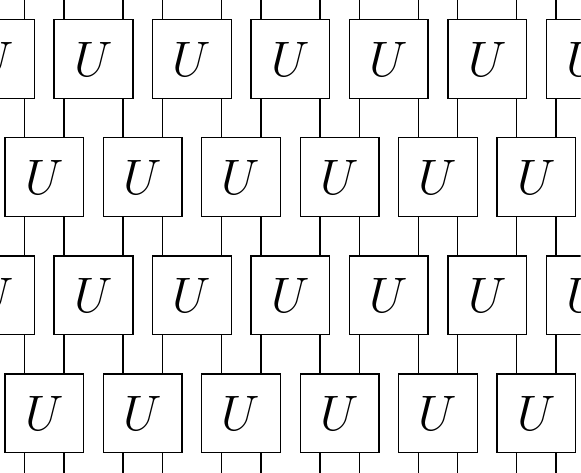}
      \label{fig:PQCA}
      \caption{Space-time diagram of a PQCA. The $U$s are the local unitary operators, the configurations are the infinite lines of wire between them.}
    \end{wrapfigure}

    Up to a constant factor, any PQCA $U$ is actually a quantum tile. 
    Therefore, $R_{m,n}\cdot U$ (or any $S\cdot U$) represents exactly a portion of space-time diagrams of the PQCA $U$, with any possible boundary condition.
    This shows how our model of quantum tile is a good generalization of Wang tiles, in the sense that it is able to represent spacetime diagrams of any PQCA, just as how classical Wang tiles can represent the space-time diagram of any radius-$\frac{1}{2}$ cellular automaton.

  \subsection{Simulating quantum walks}
    Our model of tensorial tilings exhibits typical quantum behaviors. We illustrate this by providing an example of quantum Wang tile simulating the simplest model of quantum walks on the line. The idea is to describe a walker on the line by its position $x\in \mathbb{Z}$ and its direction of movement $d\in \{\blacktriangleleft,\blacktriangleright\}$.
        In the quantum formalism the state of a walker is then a state $\ket{x,d}$, 
        \textit{i.e.} a superposition of the position in $\Z$ together with one qubit encoding a direction (left or right).
        Formally, at each step the step of the walker is of the form: $\sum_{x\in \mathbb{Z}} \sum_{d\in \{\blacktriangleleft,\blacktriangleright\}} y_{x,d}\ket{x,d}$. At each step of the walk, the walker updates its direction with a unitary $U=\begin{pmatrix}
    a & b\\ c & d
    \end{pmatrix}$ (called the coin operator), and then moves one step to the right or to the left according to its direction, which amounts to apply the shift operator $S$ such that: $\ket{x,\blacktriangleleft}\mapsto \ket{x-1,\blacktriangleleft}$ and $\ket{x,\blacktriangleright}\mapsto \ket{x+1,\blacktriangleright}$.

    So in one step the state of the walker evolves as:

    \begin{center}
    $\ket{x,\blacktriangleleft} \mapsto a\ket{x-1,\blacktriangleleft}+ c\ket{x+1,\blacktriangleright}$\\
    $\ket{x,\blacktriangleright} \mapsto b\ket{x-1,\blacktriangleleft}+ d\ket{x+1,\blacktriangleright}$
    \end{center}

    We built a quantum tile whose tilings represent the two-dimensional space time diagram of the walker evolution. We use three colors, $\emptyset$ representing an empty cell, $\blacktriangleleft$, representing a walker with left direction, and $\blacktriangleright$, representing the walker with right direction. Our quantum tile is then defined as the following tiles with the following amplitudes:
    \begin{center}
    \begin{tikzpicture}
	\begin{pgfonlayer}{nodelayer}
		\node [style=none] (0) at (-13, 1) {};
		\node [style=none] (1) at (-11, 1) {};
		\node [style=none] (2) at (-11, -1) {};
		\node [style=none] (3) at (-13, -1) {};
		\node [style=none] (4) at (-9, 1) {};
		\node [style=none] (5) at (-7, 1) {};
		\node [style=none] (6) at (-7, -1) {};
		\node [style=none] (7) at (-9, -1) {};
		\node [style=none] (8) at (-5, 1) {};
		\node [style=none] (9) at (-3, 1) {};
		\node [style=none] (10) at (-3, -1) {};
		\node [style=none] (11) at (-5, -1) {};
		\node [style=none] (12) at (-8, 0.5) {$\blacktriangleleft$};
		\node [style=none] (16) at (-1, 1) {};
		\node [style=none] (17) at (1, 1) {};
		\node [style=none] (18) at (1, -1) {};
		\node [style=none] (19) at (-1, -1) {};
		\node [style=none] (20) at (3, 1) {};
		\node [style=none] (21) at (5, 1) {};
		\node [style=none] (22) at (5, -1) {};
		\node [style=none] (23) at (3, -1) {};
		\node [style=none] (24) at (7, 1) {};
		\node [style=none] (25) at (9, 1) {};
		\node [style=none] (26) at (9, -1) {};
		\node [style=none] (27) at (7, -1) {};
		\node [style=none] (33) at (11, 1) {};
		\node [style=none] (34) at (13, 1) {};
		\node [style=none] (35) at (13, -1) {};
		\node [style=none] (36) at (11, -1) {};
		\node [style=none] (40) at (-1.5, 0) {$a$};
		\node [style=none] (42) at (2.5, 0) {$c$};
		\node [style=none] (43) at (6.5, 0) {$b$};
		\node [style=none] (44) at (10.5, 0) {$d$};
		\node [style=none] (48) at (-12, 0) {};
		\node [style=none] (49) at (-8, 0) {};
		\node [style=none] (50) at (-4, 0) {};
		\node [style=none] (51) at (12, 0) {};
		\node [style=none] (52) at (8, 0) {};
		\node [style=none] (53) at (4, 0) {};
		\node [style=none] (54) at (0, 0) {};
		\node [style=none] (55) at (-13.5, 0) {$1$};
		\node [style=none] (56) at (-9.5, 0) {$1$};
		\node [style=none] (57) at (-5.5, 0) {$1$};
		\node [style=none] (58) at (-7.5, 0) {$\blacktriangleleft$};
		\node [style=none] (59) at (-0.5, 0) {$\blacktriangleleft$};
		\node [style=none] (60) at (0, -0.5) {$\blacktriangleleft$};
		\node [style=none] (61) at (4, -0.5) {$\blacktriangleleft$};
		\node [style=none] (62) at (7.5, 0) {$\blacktriangleleft$};
		\node [style=none] (63) at (4.5, 0) {$\blacktriangleright$};
		\node [style=none] (64) at (12, -0.5) {$\blacktriangleright$};
		\node [style=none] (65) at (12.5, 0) {$\blacktriangleright$};
		\node [style=none] (66) at (8, -0.5) {$\blacktriangleright$};
		\node [style=none] (67) at (-4.5, 0) {$\blacktriangleright$};
		\node [style=none] (68) at (-4, 0.5) {$\blacktriangleright$};
	\end{pgfonlayer}
	\begin{pgfonlayer}{edgelayer}
		\draw (0.center) to (1.center);
		\draw (1.center) to (2.center);
		\draw (2.center) to (3.center);
		\draw (3.center) to (0.center);
		\draw (4.center) to (5.center);
		\draw (5.center) to (6.center);
		\draw (6.center) to (7.center);
		\draw (7.center) to (4.center);
		\draw (4.center) to (6.center);
		\draw (7.center) to (5.center);
		\draw (8.center) to (9.center);
		\draw (9.center) to (10.center);
		\draw (10.center) to (11.center);
		\draw (11.center) to (8.center);
		\draw (8.center) to (10.center);
		\draw (11.center) to (9.center);
		\draw (16.center) to (17.center);
		\draw (17.center) to (18.center);
		\draw (18.center) to (19.center);
		\draw (19.center) to (16.center);
		\draw (16.center) to (18.center);
		\draw (19.center) to (17.center);
		\draw (20.center) to (21.center);
		\draw (21.center) to (22.center);
		\draw (22.center) to (23.center);
		\draw (23.center) to (20.center);
		\draw (20.center) to (22.center);
		\draw (23.center) to (21.center);
		\draw (24.center) to (25.center);
		\draw (25.center) to (26.center);
		\draw (26.center) to (27.center);
		\draw (27.center) to (24.center);
		\draw (24.center) to (26.center);
		\draw (27.center) to (25.center);
		\draw (33.center) to (34.center);
		\draw (34.center) to (35.center);
		\draw (35.center) to (36.center);
		\draw (36.center) to (33.center);
		\draw (33.center) to (35.center);
		\draw (36.center) to (34.center);
		\draw (0.center) to (48.center);
		\draw (48.center) to (1.center);
		\draw (48.center) to (2.center);
		\draw (48.center) to (3.center);
	\end{pgfonlayer}
\end{tikzpicture}
    \end{center}
    The first three tiles propagate the walker (or its absence) encoding the shift operator $S$, and the four last encode the coin unitary $U=\begin{pmatrix}
    a & b\\ c & d
    \end{pmatrix}$. To simulate a walk with our quantum Wang tile we will consider a grid of odd width with empty boundary conditions on the left and right side. The lower side will encode the initial position of the walker with empty everywhere, except in the middle position where the walker is in a uniform superposition of direction left and right. The upper edge of the grid will allow us to read the result of the walk, we cautiously choose a width large enough to avoid forcing the walker to bounce, that is $m\geq 2n+1$. Notice that choosing $a=b=c=d=1$ we recover the classical uniform random walk, and obtain a Gaussian distribution, however taking $U=H$ we obtain the typical two-spikes distribution of the Hadamard quantum walk (see \cref{fig:chart_walk}).
    This highlight the richness of our model: a small quantum tileset is enough to simulate interesting quantum phenomena.

    \begin{figure}[ht!]
      \centering
      \vspace*{-1.5em}
      {\footnotesize \scalebox{0.8}{\input{quantumwalk.tikz}}}
      \quad
      \includegraphics[scale=0.24]{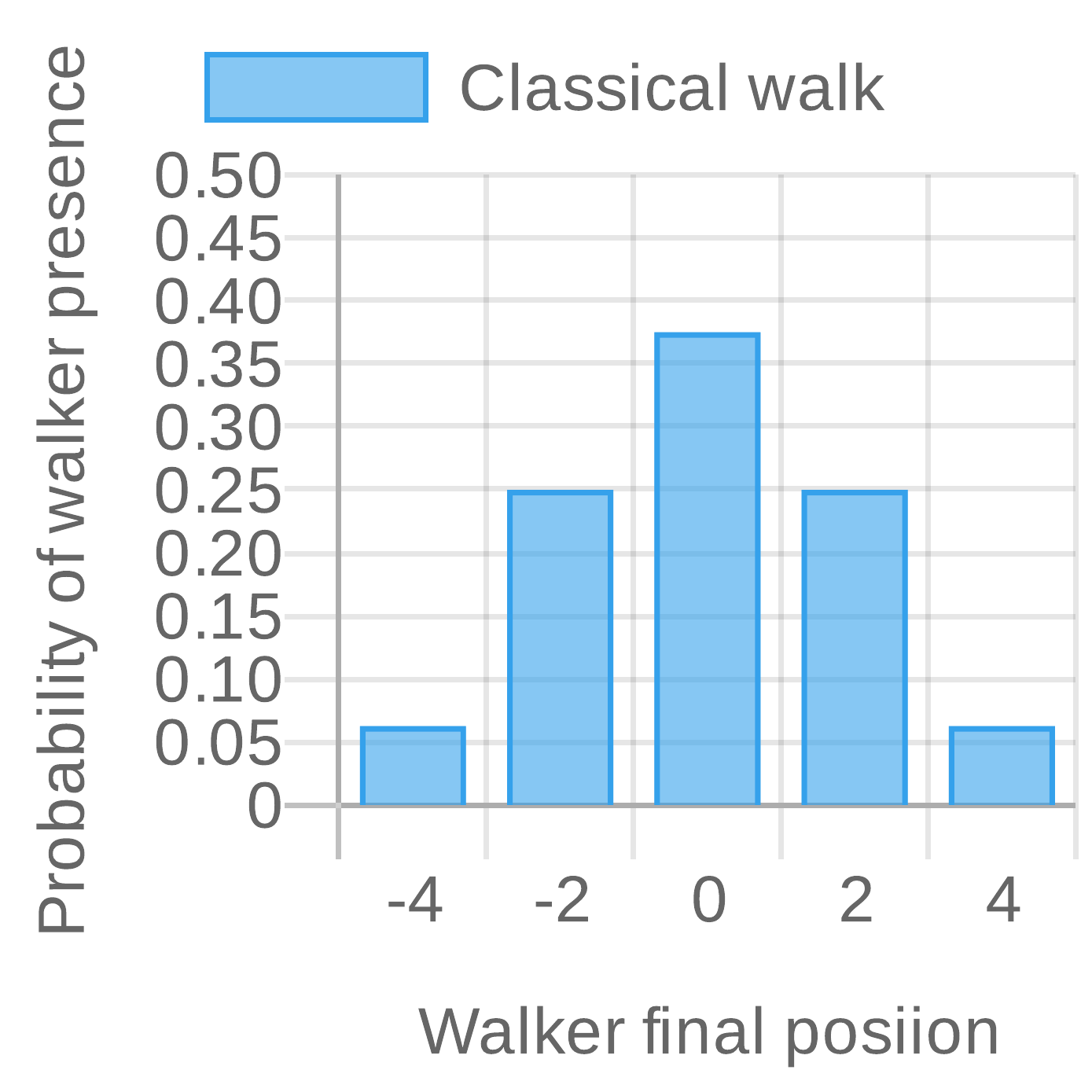}
      \includegraphics[scale=0.24]{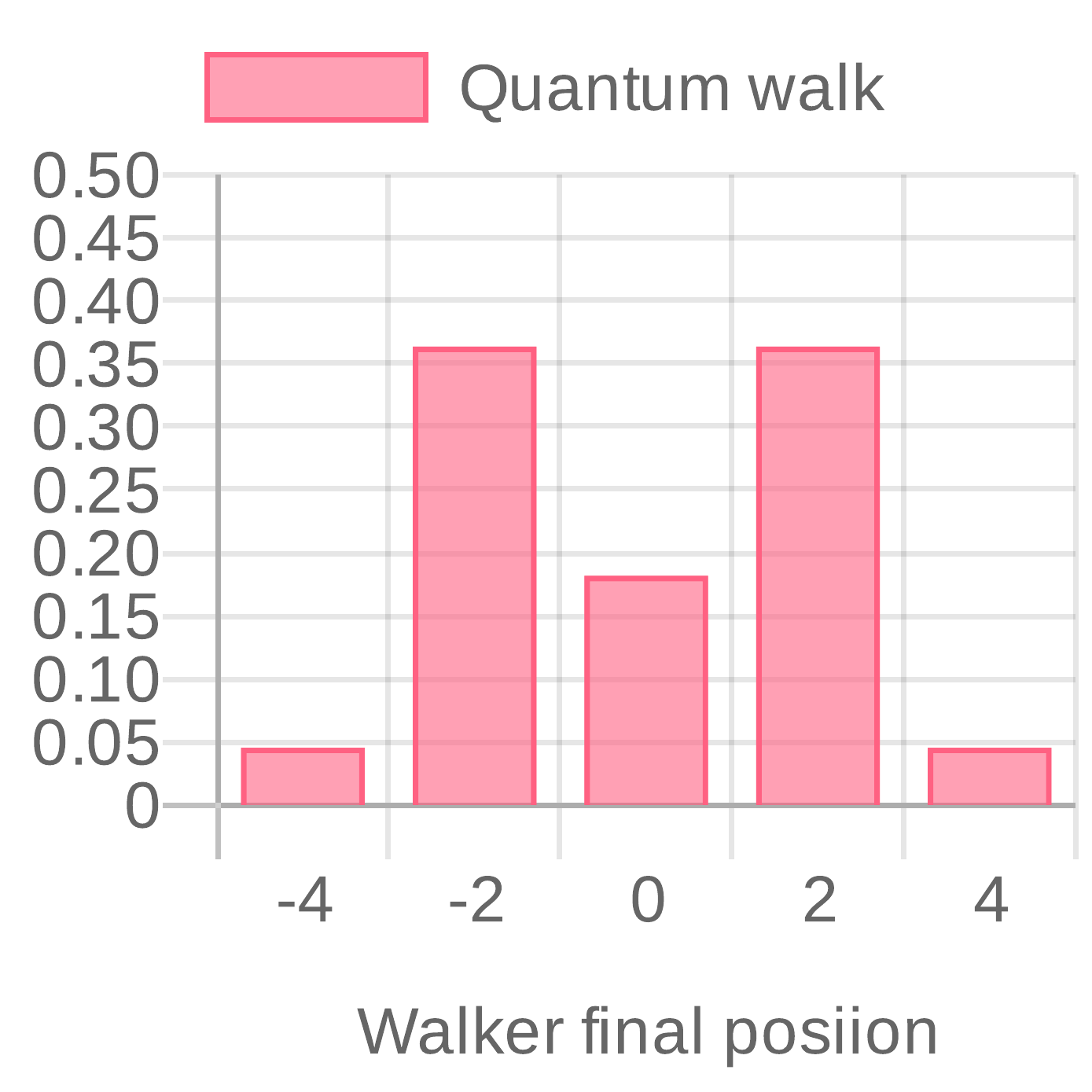}
      \caption{Simulation of four steps of a quantum walk using a $9\times 4$ pattern tiled by a quantum Wang tileset. One can see the characteristic Gaussian distribution of the classical walk (left data, in blue) and the two-spikes distribution caused by quantum interference (right data, in red). Note that the odd positions have probability 0 due to the design of the tileset.}
      \label{fig:chart_walk}
    \end{figure}

\section{Conclusion and further research directions}
  \label{sec:conclu}

	To summarize our contributions, we introduced a notion of quantum Wang tiles in dimensions one and two and demonstrated how they might behave differently than usual Wang tiles. In dimension one already, interference can change the allowed patterns of a tileset, changing the standard notion of ``tiling the line''.
  	In dimension two, the same behavior is observed with periods, even though some fundamental properties still hold.
  	This preliminary work opens a lot of questions to be investigated, and we try here to provide a list as complete as possible of relevant directions, a few of which having already been mentioned in the paper.
  
  \subsection{Minimal trace-aperiodic tiling}

  The obvious next goal is to find a ``purely quantum'' aperiodic tileset, as our current construction relies on a classical aperiodic tileset. 
  For classical Wang tilesets, Jeandel and Rao showed that no aperiodic tileset can have less than 11 tiles \cite{Jeandel_Rao_2021}. However, no such limitation is known in the quantum world, and it might be possible to use the power of interference to create a quantum aperiodic tileset with less than 11 tiles in its support.

	\subsection{Cellular automata}
	
  A key motivation for the development of quantum Wang tiles was the quest for a mathematical object playing with respect to quantum cellular automata, the role that usual Wang tiles play with respect to cellular automata.
  Our current model depicts exactly (finite) space-time diagrams of Partitioned quantum cellular automata that can simulate any other quantum cellular automata \cite{Arrighi_Grattage_2012}. 
  This connection is not developed in the present paper and is the subject of ongoing work. We are particularly interested in the relationship between determinism in tilings and causality in tensor networks.
		
  \subsection{Condensed matter physics}

  We expect that linear algebraic and string diagrammatic methods could be of interest to reformulate and to study quantitative aspects of tiling theory. The properties of possibilistic tiles allow to count the number of admissible patterns via algebraic methods, a very important problem in statistical physics (see for example \cite{kasteleyn1961statistics}). More generaly we hope our model could act as a bridge between symbolic dynamic and condensed matter theory. The object studied in those field are very similar but are approached with completely different sets of question by the different comunities. Indeed diagrams similar to ours have also been used in condensed matter theory, for example, to represent MPS product states \cite{orus2014practical}. The precise link between our framework and this field will be the object of further work. The cornerstone question here being to link the physicist point of view on local constraints, describing the desired configuration as the ground state of a Hamiltonian, with the approach of mathematicians and computer scientists, who typically understand local constraints as local forbidden patterns.

\subsection{Diagrammatical rewriting}
	
	One can have a purely syntactic approach to tensors using diagrammatic equational theories. In a sense, this mean opening the tensorial tiles to build them from elementary generators and apply rewriting technics to study there properties. This approach allows a direct connection with the diagrammatic languages used to represent quantum processes, as $ZX$, $ZW$ and $ZH$-calculus \cite{coecke2008interacting,coecke2010compositional,backens2018zh} that have already been used to approach similar combinatorial problems \cite{de2020tensor}. The diagrammatical language share many property with opur model, in particular they share the same freedom in defining prefered a prefered direction of time to interprte the diagram. More generally, even if we choose not to emphasize this aspect too much, our approach is deeply rooted in the categorical quantum mechanic program, and the theory of monoidal categories \cite{bob2017picturing}. Indeed the way we defined tensorial tiles can be naturally generalized in compact closed categories.

\subsection{Diagrammatical symbolic dynamics}

	A very powerful tool to understand classical Wang tilings is the strong link with subshifts of finite type and thus symbolic dynamics. It would be interesting to design a quantum analog of subshifts of finite type. However, this is not a straightforward task, mainly because the no-cloning theorem prevents cells from communicating their whole state to more than one neighbor. 

  We introduced quantum Wang tiles only for the line and the plane but the extension in higher dimension is straightforward. Hypercube tensorial tiles of dimension $d$ are tensors with $2d$ indices and definitions of possibilistic, probabilistic and quantum hyper tiles are completely analogous. The formalism also directly extends to the case of other shapes, such as triangles or hexagons, extensions to different topologies are also possible, like a torus or Cayley graphs of groups. 

  Our ongoing work on reformulating symbolic dynamics in a tensorial formalism raises the hope that it allows natural generalization of one-dimensional techniques (for example based on finite automata) to study configurations on more exotic topologies than the line.

 \subsection{Skolem}
 
  Our tensorial view on tilings allow us to draw a new bridge between tilings and the Skolem problem \cite{skolem1}.
  The Skolem problem can be formulated as follow. The input is a square integer matrix $M$ with a pair of vectors $\ket a , \ket b$. It asks the question whether there exists $n$ such that $\bra a M^n \ket b = 0$.
  %
  It is in fact very similar to the problem of tiling the line with tensorial dominoes. Indeed, tensorial domino $M$ does \emph{not} tile the line if and only if there exists $n$ such that for all pair of vectors $\ket a , \ket b$, $\bra a M^n \ket b = 0$.
  However, one can notice the difference of quantification between the two problems, leading to different results regarding their decidability: the tiling problem is decidable, as equivalent to deciding the nilpotency of the matrix, whereas the decidability of the Skolem problem is still unknown.
  This similarity allow us to interpret the Skolem problem as a (tensorial) tiling problem: there is no zero in the sequence $(\bra a M^n \ket b)_n$ if and only if the tensorial tile $M$ admits arbitrary wide patterns with border $\bra a, \ket b$. In other words, the Skolem problem corresponds to a "constrained" domino problem, where the border is imposed to be $\bra a, \ket b$ instead of any color vector.

  It is also possible to reduce the Skolem problem to the two-dimensional tensorial tiling problem using a Robinson aperiodic tileset, and classical tricks from tilings in order to make patterns $\bra a M^n \ket b$ appear in the tiling for all $n$, in the squares drawn by the Robinson tileset. This tensorial tile then tiles the plane if and only if $\bra a M^n \ket b \neq 0$ for all $n$.
  Of course the 2D tiling problem being undecidable, this does not show help deciding Skolem problem, but this draws an interesting link between Skolem's problem and tiling problems. It places it somewhere at the fascinating frontier between the 1D tiling problem (decidable) and 2D tiling problem (undecidable).

  This link also shows that our open problems might be very difficult to tackle. Especially \cref{open1}, as it can be understood as understanding the zeroes of a sequence defined not by a matrix product as in Skolem, but by more general 2D tensor contraction. If Skolem's problem is not fully understood for matrix power, \cref{open1} looks even more challenging as the "growing square of a tensor" operation is much more complex and less understood than matrix power.

\section*{Acknowledgements}
  The authors warmly thank Pablo Arrighi for the many fruitful discussions and funding at the early stage of this project, as well as Emmanuel Jeandel for his useful discussions.

  This publication was made possible through the support of the ID\# 61466 grant from the John Templeton Foundation, as part of the “The Quantum Information Structure of Spacetime (QISS)” Project (\url{qiss.fr}). The opinions expressed in this publication are those of the author(s) and do not necessarily reflect the views of the John Templeton Foundation.

\bibliographystyle{plainurl}
\bibliography{biblio}

\end{document}